\documentclass[11pt]{llncs}
\usepackage{tabularx,booktabs,multirow,delarray,array}
\usepackage{graphicx,amsmath,amssymb}
\usepackage{latexsym}
\usepackage[utf8]{inputenc}

\usepackage[in]{fullpage}

\def\calP{\mathcal{P}}

\def\calA{\mathcal{A}}

\def\calT{\mathcal{T}}
\def\bbR{\mathbb{R}}

\def\Ed{\mathsf{E}\mathrm{d}}

\begin{document}

\title{Covering Uncertain Points in a Tree\thanks{A preliminary version of this paper will appear in the Proceedings of the 15th Algorithms and Data Structures Symposium (WADS 2017). This research was supported in part by NSF under Grant CCF-1317143.}}

\author{Haitao Wang
\and
Jingru Zhang
}

\institute{
Department of Computer Science\\
Utah State University, Logan, UT 84322, USA\\
\email{haitao.wang@usu.edu,jingruzhang@aggiemail.usu.edu}\\
}

\maketitle

\pagestyle{plain}
\pagenumbering{arabic}
\setcounter{page}{1}

\begin{abstract}
In this paper, we consider a coverage problem for uncertain points in a tree.
Let $T$ be a tree containing a set $\mathcal{P}$
of $n$ (weighted) demand points, and the location of each demand point
$P_i\in \calP$ is uncertain but is known to appear in one of $m_i$
points on $T$ each associated with a probability. Given a {\em
covering range $\lambda$}, the problem is to find a minimum number of
points (called {\em centers}) on $T$ to build facilities for serving
(or covering) these demand points in the sense that for each uncertain
point $P_i\in \calP$, the expected distance from $P_i$ to at least one
center is no more than $\lambda$. The problem has not been studied before.
We present an $O(|T|+M\log^2 M)$ time algorithm for the problem, where $|T|$ is the number of vertices of $T$ and
$M$ is the total number of locations of all uncertain points of $\calP$, i.e.,
$M=\sum_{P_i\in \calP}m_i$. In addition, by using this algorithm, we
solve a $k$-center problem on $T$ for the uncertain points of
$\calP$.
\end{abstract}


\section{Introduction}
\label{sec:Gintro}


Data uncertainty is very common in many applications,
such as sensor databases, image resolution, facility location services, and it is mainly due to measurement inaccuracy, sampling discrepancy, outdated data sources, resource limitation, etc.
Problems on uncertain data have attracted considerable attention, e.g.,
\cite{ref:AgarwalIn09,ref:AgarwalNe12,ref:AgarwalCo14,ref:ChengCl08,ref:ChengEf04,ref:DongDa07,ref:KamousiCl11,ref:KamousiSt11,ref:SuriOn14,ref:SuriOn13,ref:TaoRa07}.
In this paper, we study a problem of covering uncertain points on a tree. The problem is formally defined as follows.


Let $T$ be a tree. We consider each edge $e$ of $T$ as a line segment of a positive length
so that we can talk about ``points'' on $e$. Formally, we specify a point $x$ of $T$ by an edge $e$ of $T$ that contains $x$ and the distance between $x$ and an incident vertex of $e$. The distance of any two
points $p$ and $q$ on $T$, denoted by $d(p,q)$, is defined as the sum of the lengths of all edges on the simple path from $p$ to $q$ in $T$. Let
$\calP=\{P_1,\ldots,P_n\}$ be a set of $n$ uncertain (demand) points
on $T$. Each $P_i\in \calP$ has $m_i$ possible locations on $T$,
denoted by $ \{p_{i1}, p_{i2}, \cdots, p_{im_i}\} $, and each location $p_{ij}$
of $P_i$ is associated with a probability $f_{ij}\geq 0$ for $P_i$
appearing at $p_{ij}$ (which is independent of other locations),
with $\sum_{j=1}^{m_i}f_{ij}=1$; e.g.,
see Fig.~\ref{fig:uncertainty}.
In addition, each $P_i\in \calP$ has a weight $w_i\geq 0$. For any point $x$ on $T$, the (weighted) {\em expected distance} from $x$ to $P_i$,
denoted by $\Ed(x,P_i)$, is defined as
$$\Ed(x, P_i)=w_i\cdot\sum_{j=1}^{m_i}f_{ij}\cdot d(x, p_{ij}).$$


\begin{figure}[t]
\begin{minipage}[t]{\linewidth}
\begin{center}
\includegraphics[totalheight=1.0in]{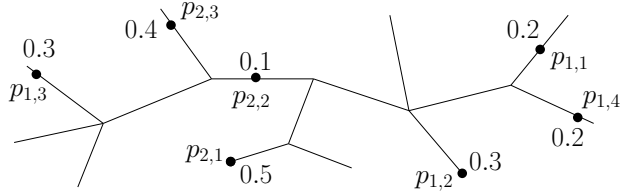}
\caption{\footnotesize Illustrating two uncertain points $P_1$ and
$P_2$, where $P_1$ has four possible locations and $P_2$ has three
possible locations.
The numbers are the probabilities.}
\label{fig:uncertainty}
\end{center}
\end{minipage}
\vspace*{-0.15in}
\end{figure}

Given a value $\lambda\geq 0$, called the {\em covering range}, we say that a point $x$ on $T$ {\em covers} an uncertain point $P_i$ if $\Ed(x,P_i)\leq \lambda$.
The {\em center-coverage problem} is to compute a minimum number of points on $T$, called {\em centers}, such that every uncertain point of $\calP$ is covered by at least one center (hence we can build facilities on these centers to ``serve'' all demand points).

To the best of our knowledge, the problem has not been studied before.
Let $M$ denote the total number of locations all uncertain points, i.e., $M=\sum_{i=1}^n m_i$. Let $|T|$ be the number of vertices of $T$.
In this paper, we present an algorithm that solves the problem in $O(|T|+M\log^2M)$ time, which is nearly linear as the input size of the problem is $\Theta(|T|+M)$.

As an application of our algorithm, we also solve a dual problem,
called the {\em $k$-center} problem, which is to compute a
number of $k$ centers on $T$ such that the covering range is
minimized. Our algorithm solves the $k$-center problem in
$O(|T|+n^2\log n\log M+M\log^2M\log n)$ time.


\subsection{Related Work}

Two models on uncertain data have been commonly considered:  the \textit{existential} model
\cite{ref:AgarwalCo14,ref:KamousiCl11,ref:KamousiSt11,ref:SuriOn14,ref:SuriOn13,ref:YiuEf09}
and the \textit{locational} model
\cite{ref:AgarwalIn09,ref:AgarwalNe12,ref:ChengEf04,ref:TaoRa07}.
In the existential
model an uncertain point has a specific location but its existence is uncertain while
in the locational model an uncertain point always exists but its location is uncertain and follows a probability distribution function.
Our problems belong to the locational model. In fact, the same
problems under existential model are essentially the weighted case
for ``deterministic'' points (i.e., each $P_i\in \calP$ has a
single ``certain'' location), and the center-coverage problem is solvable in linear time~\cite{ref:KarivAnC79} and the $k$-center problem is solvable in $O(n\log^2 n)$ time~\cite{ref:ColeSl87,ref:MegiddoNe83}.


If $T$ is a path, both the center-coverage problem and the $k$-center
problem on uncertain points have been studied~\cite{ref:WangOn15}, but
under a somewhat special problem setting where $m_i$ is the same for
all $1\leq i\leq n$.
The two problems were solved in
$O(M+n\log k)$ and $O(M\log M + n \log k \log n)$ time, respectively.
If $T$ is tree, an $O(|T|+M)$ time algorithm was given in
\cite{ref:WangCo16} for the one-center problem under the above special
problem setting.

As mentioned above, the ``deterministic'' version of the center-coverage problem is solvable in linear
time~\cite{ref:KarivAnC79}, where all demand points are on the vertices.
For the $k$-center problem, Megiddo and Tamir
\cite{ref:MegiddoNe83} presented an $O(n\log^2 n\log\log n)$ time
algorithm ($n$ is the size of the tree), which was improved to $O(n\log^2 n)$ time
by Cole~\cite{ref:ColeSl87}. The unweighted case was solved in linear time by Frederickson
\cite{ref:FredericksonPa91}.

Very recently, Li and Huang~\cite{ref:HuangSt17} considered the same $k$-center problem under the same uncertain model as ours but in the Euclidean space, and they gave an approximation algorithm.
Facility location problems in other uncertain models have also been considered.
For example, L{\"o}ffler and van Kreveld \cite{ref:LofflerLa10} gave
algorithms for
computing the smallest enclosing circle for imprecise points
each of which is contained in a planar region (e.g., a circle or a square).
J{\o}rgenson et al. \cite{ref:JorgensenGe11}
studied the problem of computing the distribution of
the radius of the smallest enclosing circle for uncertain points
each of which has multiple locations in the plane.
de~Berg et al. \cite{ref:BergKi13} proposed algorithms for dynamically
maintaining
Euclidean $2$-centers for a set of moving points in the plane (the
moving points are considered uncertain).
See also the problems for minimizing the maximum regret, e.g.,
\cite{ref:AverbakhFa05,ref:AverbakhMi97,ref:WangMi14}.

Some coverage problems in various geometric settings have also been studied.
For example,
the unit disk coverage problem is to compute a minimum number of unit
disks to cover a given set of points in the plane. The problem is
NP-hard and a polynomial-time approximation scheme
was known~\cite{ref:HochbaumAp85}. The discrete case where the disks must
be selected from a given set was also studied~\cite{ref:MustafaPT09}.
See~\cite{ref:BeregOp15,ref:ChanGe15,ref:GonzalezCo91,ref:KimCo11}
and the references therein for various problems of covering points using squares.
Refer to a survey~\cite{ref:AgarwalEf98} for more geometric coverage problems.

\subsection{Our Techniques}
We first discuss our techniques for solving the center-coverage
problem.

For each uncertain point $P_i\in \calP$, we find a point $p^*_i$ on $T$
that minimizes the expected distance $\Ed(p_i,P_i)$, and $p^*_i$
is actually the weighted median of all locations of $P_i$. We observe
that if we move a point $x$ on $T$ away from
$p_i^*$, the expected distance $\Ed(x,P_i)$ is monotonically increasing. We
compute the medians $p_i^*$ for all uncertain points in $O(M\log M)$ time. Then we show
that there exists an optimal solution in which all centers are in $T_m$,
where $T_m$ is the minimum subtree of $T$ that connects all medians
$p_i^*$ (so every leaf of $T_m$ is a median $p_i^*$).
Next we find centers on $T_m$. To this end, we propose a simple
greedy algorithm, but the challenge is on developing efficient data structures
to perform certain operations. We briefly discuss it below.

We pick an arbitrary vertex $r$ of $T_m$ as the root.
Starting from the leaves, we consider the vertices of $T_m$ in
a bottom-up manner and place centers whenever we ``have to''. For
example, consider a leaf $v$ holding a median $p_i^*$ and let $u$ be
the parent of $v$. If $\Ed(u,P_i)> \lambda$, then we have to place a
center $c$ on the edge $e(u,v)$ in order to cover $P_i$. The location
of $c$ is chosen to be at a point of $e(u,v)$ with
$\Ed(c,P_i)=\lambda$ (i.e., on the one hand, $c$ covers $P_i$, and on
the other hand, $c$ is as close to $u$ as possible in the
hope of covering other uncertain points as many as possible). After
$c$ is placed, we find and remove all uncertain points
that are covered by $c$. Performing this operation efficiently is a key difficulty for our approach.
We solve the problem in an output-sensitive manner
by proposing a dynamic data structure that also
supports the remove operations.

We also develop data structures for other operations needed in
the algorithm. For example, we build a data structure in $O(M\log M)$ time that can compute the expected distance $\Ed(x,P_i)$ in $O(\log M)$
time for any point $x$ on $T$ and any $P_i\in \calP$.
These data structures may be of independent interest.

We should point out that our algorithm is essentially different from the one in our previous work \cite{ref:WangCo16}. Indeed, our algorithm here is a greedy algorithm while the one in \cite{ref:WangCo16} uses the prune-and-search technique. Also our algorithm relies heavily on some data structures as mentioned above while the algorithm in \cite{ref:WangCo16} does not need any of these data structures.

For solving the $k$-center problem, by observations, we first
identify a set of $O(n^2)$ ``candidate'' values
such that the covering range in the optimal solution must be in the
set. Subsequently, we use our algorithm for the center-coverage
problem as a decision procedure to find the optimal covering
range in the set.

Note that although we have assumed $\sum_{j=1}^{m_i}f_{ij} =
1$ for each $P_i\in \calP$, it is quite straightforward to adapt our algorithm
to the general case where the assumption does not hold.

The rest of the paper is organized as follows.
We introduce some notation in Section~\ref{sec:pre}.
In Section~\ref{sec:main},
we describe our algorithmic scheme for the center-coverage problem but
leave the implementation details in the subsequent two sections.
Specifically, the algorithm for computing all medians $p_i^*$ is given
in Section~\ref{sec:decmedian}, and in the same section  we also
propose a connector-bounded centroid
decomposition of $T$, which is repeatedly used in the paper and may
be interesting in its own right.
The data structures used in our algorithmic scheme are given in Section~\ref{sec:ds}.
We finally solve the $k$-center problem in Section~\ref{sec:kcenter}.

\section{Preliminaries}
\label{sec:pre}

Note that the locations of the uncertain points of $\calP$ may be in
the interior of the edges of $T$.  A
{\em vertex-constrained case} happens if all locations of $\calP$ are at
vertices of $T$ and each vertex of $T$ holds at least one location of
$\calP$ (but the centers we seek can still be in the interior of edges).
As in \cite{ref:WangCo16}, we will show later in Section~\ref{sec:reduction} that
the general problem can be reduced to the vertex-constrained case in $O(|T|+M)$
time. In the following, unless otherwise stated, we focus our
discussion on the vertex-constrained case and assume our problem on $\calP$ and $T$ is a vertex-constrained case.
For ease of exposition, we further make a general position assumption that every
vertex of $T$ has only one location of $\calP$ (we explain in Section~\ref{sec:reduction} that our algorithm easily extends to the degenerate case). Under this assumption,
it holds that $|T|=M\geq n$.

Let $e(u,v)$ denote the edge of $T$ incident to two vertices $u$ and $v$.
For any two points $p$ and $q$ on $T$, denote by $\pi(p,q)$ the simple
path from $p$ to $q$ on $T$.

Let $\pi$ be any simple path on $T$ and $x$ be any point on $\pi$.
For any location $p_{ij}$ of an uncertain point $P_i$, the distance
$d(x,p_{ij})$ is a convex (and piecewise linear) function as $x$ changes on
$\pi$ \cite{ref:MegiddoLi83}.
As a sum of multiple convex functions,
$\Ed(x,P_i)$ is also convex (and piecewise linear) on
$\pi$, that is, in general, as $x$ moves
on $\pi$, $\Ed(x,P_i)$ first monotonically decreases and then monotonically increases.
In particular, for each edge $e$ of $T$, $\Ed(x,P_i)$ is a linear
function for $x\in e$.

For any subtree $T'$ of $T$ and any
$P_i\in \calP$, we call the sum of the probabilities of the locations
of $P_i$ in $T'$ the {\em probability sum} of $P_i$ in $T'$.

For each uncertain point $P_i$, let $p_i^*$ be a point $x\in T$ that minimizes
$\Ed(x,P_i)$. If we consider $w_i\cdot f_{ij}$ as the
weight of $p_{ij}$, $p_i^*$ is actually the {\em weighted median} of
all points $p_{ij}\in P_i$. We call $p_i^*$ the {\em median} of $P_i$.
Although $p_i^*$ may not be unique (e.g.,
when there is an edge $e$ dividing $T$ into two subtrees such that the probability sum
of $P_i$ in either subtree is exactly $0.5$), $P_i$
always has a median located at a vertex $v$ of $T$, and we let $p_i^*$
refer to such a vertex.

Recall that $\lambda$ is the given covering range for the center-coverage problem.
If $\Ed(p_i^*,P_i)>\lambda$ for some $i\in [1,n]$, then
there is no solution for the problem since no point of $T$ can cover $P_i$. Henceforth, we assume
$\Ed(p_i^*,P_i)\leq \lambda$ for each $i\in [1,n]$.


\section{The Algorithmic Scheme}
\label{sec:main}
In this section, we describe our algorithmic scheme for the center-coverage problem,  and
the implementation details will be presented in the subsequent two
sections.

We start with computing the medians $p_i^*$ of all uncertain points of
$\calP$. We have the following lemma, whose proof is deferred to
Section~\ref{sec:median}.

\begin{lemma}\label{lem:median}
The medians $p_i^*$ of all uncertain points $P_i$ of $\calP$ can be
computed in $O(M\log M)$ time.
\end{lemma}

\subsection{The Medians-Spanning Tree $T_m$}
\label{sec:mediantree}

Denote by $P^*$ the set of all medians $p_i^*$.
Let $T_m$ be the minimum connected subtree of $T$ that spans/connects all
medians. Note that each leaf of $T_m$ must hold a median.
We pick an
arbitrary median as the root of $T$, denoted by $r$. The
subtree $T_m$ can be easily computed in $O(M)$ time by a
post-order traversal on $T$ (with respect to the root $r$), and we
omit the details. The following lemma is based on the fact that $\Ed(x,P_i)$ is
convex for $x$ on any simple path of $T$ and $\Ed(x,P_i)$ minimizes at
$x=p_i^*$.

\begin{lemma}\label{lem:Tm}
There exists an optimal solution for the center-coverage problem in which every
center is on $T_m$.
\end{lemma}
\begin{proof}
Consider an optimal solution and let $C$ be the set of all centers
in it. Assume there is a center $c\in C$ that is not on $T_m$.
Let $v$ be the vertex of $T_m$ that holds a median and is closest to $c$.
Then $v$ decomposes $T$ into two subtrees $T_1$ and $T_2$ with the only
common vertex $v$ such that $c$ is in one subtree, say $T_1$, and all
medians are in $T_2$. If
we move a point $x$ from $c$ to $v$ along $\pi(c,v)$, then
$\Ed(x,P_i)$ is non-increasing for each $i\in [1,n]$. This implies that if we move
the center $c$ to $v$, we can obtain an optimal solution in which $c$ is
in $T_m$.

If $C$ has other centers that are not on $T_m$, we do the
same as above to obtain an optimal solution in which all centers are
on $T_m$. The lemma thus follows.
\qed
\end{proof}

Due to Lemma~\ref{lem:Tm}, we will focus on finding centers on $T_m$. We also
consider $r$ as the root of $T_m$. With respect to $r$, we can talk
about ancestors and descendants of the vertices in $T_m$.
Note that for any two vertices $u$ and $v$ of $T_m$, $\pi(u,v)$ is in $T_m$.

We reindex all medians and the corresponding
uncertain points so that the new indices will facilitate our
algorithm, as follows.
Starting from an arbitrary child of $r$ in $T_m$, we traverse down
the tree $T_m$ by always following the leftmost child of the current
node until we encounter a leaf, denoted by $v^*$. Starting from $v^*$ (i.e., $v^*$ is the first visited leaf),
we perform a post-order traversal on $T_m$ and reindex all medians of
$P^*$ such that $p_1^*,p_2^*,\ldots,p^*_n$ is the list of points of
$P^*$ visited in order in the above traversal. Recall that the
root $r$ contains a median, which is $p^*_n$ after the reindexing. Accordingly, we also
reindex all uncertain points of $\calP$ and their corresponding
locations on $T$, which can be done in $O(M)$ time. In the following
paper, we will always use the new indices.

 For each vertex $v$ of $T_m$, we
use $T_m(v)$ to represent the subtree of $T_m$ rooted at $v$. The reason we do the
above reindexing is that for any vertex $v$ of $T_m$, the new indices of all
medians in $T_m(v)$ must form a range $[i,j]$ for
some $1\leq i\leq j\leq n$, and we use $R(v)$ to denote the range. It will be clear later that this property will facilitate our algorithm.


\subsection{The Algorithm}

Our algorithm for the center-coverage problem works as follows. Initially,
all uncertain points are ``active''. During the algorithm, we will
place centers on $T_m$, and once an uncertain point $P_i$ is covered
by a center, we will ``deactivate'' it (it then becomes
``inactive''). The algorithm visits all
vertices of $T_m$ following the above post-order
traversal of $T_m$ starting from leaf $v^*$. Suppose $v$ is currently being
visited. Unless $v$ is the root $r$, let $u$ be the parent of $v$.
Below we describe our algorithm for processing $v$. There are two cases
depending on whether $v$ is a leaf or an internal node, although the
algorithm for them is essentially the same.

\subsubsection{The Leaf Case}
If $v$ is a leaf, then it holds a median $p_i^*$. If $P_i$ is inactive,
we do nothing; otherwise, we proceed as follows.

We compute a point $c$ (called a {\em candidate center}) on the path
$\pi(v,r)$ closest to $r$ such that
$\Ed(c,P_i)\leq \lambda$. Note that if we move a point $x$ from $v$ to $r$ along $\pi(v,r)$,
$\Ed(x,P_i)$ is monotonically increasing. By the definition of $c$,
if $\Ed(r,P_i)\leq \lambda$, then $c=r$; otherwise, $\Ed(c,P_i)=\lambda$.
If $c$ is in $\pi(u,r)$, then we do nothing and finish processing $v$.
Below we assume that $c$ is not in $\pi(u,r)$ and thus is in
$e(u,v)\setminus\{u\}$ (i.e., $c\in e(u,v)$ but $c\neq u$).

In order to cover $P_i$, by the definition of $c$, we must place a
center in $e(u,v)\setminus\{u\}$. Our strategy is to place a center at $c$.
Indeed, this is the
best location for placing a center since it is the location that can
cover $P_i$ and is closest to $u$ (and thus is closest to every other
active uncertain point).
We use a {\em candidate-center-query} to compute $c$ in
$O(\log n)$ time,  whose details will be discussed later.
Next, we report all active uncertain points that can be covered by $c$,
and this is done by a {\em coverage-report-query} in output-sensitive
$O(\log M\log n+k\log n)$ amortized time, where $k$ is the number of uncertain points covered by $c$. The details for the operation will be discussed later.
Further, we deactivate all these uncertain points.
We will show that deactivating each uncertain point $P_j$ can be done
in $O(m_j\log M\log n)$ amortized time.
This finishes processing $v$.

\subsubsection{The Internal Node Case}
If $v$ is an internal node, since we process the vertices of $T_m$
following a post-order traversal, all descendants of $v$ have already been processed. Our algorithm
maintains an invariant that if the subtree $T_m(v)$ contains any active median
$p_i^*$ (i.e., $P_i$ is active), then $\Ed(v,P_i)\leq \lambda$.
When $v$ is a leaf, this invariant trivially holds. Our way of
processing a leaf discussed above also maintains this invariant.

To process $v$, we first check whether $T_m(v)$ has any active medians.
This is done by a
{\em range-status-query} in $O(\log n)$ time, whose details
will be given later.
If  $T_m(v)$  does not have any active median, then we are done with
processing $v$.
Otherwise, by the algorithm invariant, for each active median
$p_i^*$ in $T_m(v)$, it holds that $\Ed(v,P_i)\leq \lambda$.
If $v=r$, we place a center at $v$ and finish the entire algorithm.
Below, we assume $v$ is not $r$ and thus $u$ is the parent of $v$.

We compute a point $c$ on $\pi(v,r)$ closest to $r$ such that
$\Ed(c,P_i)\leq \lambda$ for all active medians $p_i^*\in T_m(v)$, and we call
$c$ the {\em candidate center}.
By the definition of $c$, if $\Ed(r,P_i)\leq \lambda$ for all active
medians $p_i^*\in T_m(v)$, then $c=r$; otherwise, $\Ed(c,P_i)=\lambda$ for
some active median $p_i^*\in T_m(v)$. As in the leaf case,
finding $c$ is done in $O(\log n)$ time by a {\em candidate-center-query}.
If $c$ is on $ \pi(u,r)$, then we finish processing $v$.
Note that this implies $\Ed(u,P_i)\leq \lambda$ for each active
median $p_i^*\in T_m(v)$, which maintains the algorithm invariant for $u$.

If $c\not\in \pi(u,r)$, then $c\in e(u,v)\setminus\{u\}$.
In this case, by the definition of $c$,
we must place a center in $e(u,v)\setminus\{u\}$ to cover $P_i$.
As discussed in the leaf case, the best location for placing a center
is $c$ and thus we place a center at $c$.
Then, by using a coverage-report-query, we find all active uncertain points
covered by $c$ and deactivate them. Note that by the definition of $c$, $c$ covers $P_j$ for all medians $p_j^*\in T_m(v)$.
This finishes processing $v$.

Once the root $r$ is processed, the algorithm finishes.

\subsection{The Time Complexity}

To analyze the running time of the algorithm, it remains to discuss the three operations: range-status-queries, coverage-report-queries, and candidate-center-queries. For answering range-status-queries, it is trivial, as shown in Lemma~\ref{lem:status}.

\begin{lemma}\label{lem:status}
We can build a data structure in $O(M)$ time that can answer each range-status-query in $O(\log n)$ time. Further, once an uncertain point is deactivated, we can remove it from the data structure in $O(\log n)$ time.
\end{lemma}
\begin{proof}
Initially we build a balanced binary search tree $\Phi$ to maintain all indices $1,2,\ldots,n$. If an uncertain point $P_i$ is deactivated, then we simply remove $i$ from the tree in $O(\log n)$ time.

For each range-status-query, we are given a vertex $v$ of $T_m$, and the goal is to decide whether $T_m(v)$ has any active medians. Recall that all medians in $T_m(v)$ form a range $R(v)=[i,j]$. As preprocessing, we compute $R(v)$ for all vertices $v$ of $T_m$, which can be done in $O(|T_m|)$ time by the post-order traversal of $T_m$ starting from leaf $v^*$. Note that $|T_m|=O(M)$.

During the query, we simply check whether $\Phi$ still contains any index in the range $R(v)=[i,j]$, which can be done in $O(\log n)$ time by standard approaches (e.g., by finding the successor of $i$ in $\Phi$).
\qed
\end{proof}

For answering the coverage-report-queries and the candidate-center-queries, we have the following two lemmas. Their proofs are deferred to Section~\ref{sec:ds}.

\begin{lemma}\label{lem:report}
We can build a data structure $\calA_1$ 
in $O(M\log^2 M)$ time that can answer in $O(\log M \log n+k\log n)$ amortized time each
coverage-report-query, i.e., given any point $x\in T$, report
all active uncertain points covered by $x$, where $k$ is the output size.
Further, if an uncertain
point $P_i$ is deactivated, we can remove $P_i$ from $\calA_1$
in $O(m_i\cdot \log M\cdot \log n)$ amortized time.
\end{lemma}


\begin{lemma}\label{lem:candidate}
We can build a data structure $\calA_2$ in
$O(M\log M+n\log^2 M)$ time that can answer in $O(\log n)$ time each
candidate-center-query, i.e., given any vertex $v\in T_m$,
find the candidate center $c$ for the active medians of $T_m(v)$.
Further, if an uncertain point $P_i$ is deactivated, we can remove $P_i$ from $\calA_2$ in $O(\log n)$ time.
\end{lemma}

Using these results, we obtain the following.

\begin{theorem}\label{theo:10}
We can find a minimum number of centers on $T$ to cover all uncertain
points of $\calP$ in $O(M\log^2 M)$ time.
\end{theorem}
\begin{proof}
First of all, the total preprocessing time of Lemmas~\ref{lem:status}, \ref{lem:report}, and \ref{lem:candidate} is $O(M\log^2 M)$.
Computing all medians takes $O(M\log M)$ time by Lemma~\ref{lem:median}.
Below we analyze the total time for computing centers on $T_m$.

The algorithm processes each vertex of $T_m$ exactly once. The
processing of each vertex calls each of the following three operations
at most once: coverage-report-queries, range-status-queries, and candidate-center-queries.
Since each of the last two operations runs in $O(\log n)$ time, the
total time of these two operations in the entire algorithm is $O(M\log
n)$. For the coverage-report-queries, each operation runs in $O(\log M \log n+k\log n)$
amortized time. Once an uncertain point $P_i$ is reported by it, $P_i$ will
be deactivated by removing it from all three data structures (i.e., those in  Lemmas~\ref{lem:status}, \ref{lem:report}, and \ref{lem:candidate}) and $P_i$
will not become active again. Therefore, each uncertain point will be
reported by the coverage-report-query operations at most once. Hence, the
total sum of the value $k$ in the entire algorithm is $n$. Further, notice that there are at most $n$ centers placed by the algorithm. Hence, there are at most $n$ coverage-report-query operations in the algorithm. Therefore, the
total time of the coverage-report-queries in the entire
algorithm is $O(n\log M\log n)$. In addition, since
each uncertain point $P_i$ will be deactivated at
most once, the total time of the remove operations for all three data
structures in the entire algorithm is $O(M\log M\log n)$ time.

As $n\leq M$, the theorem follows.  \qed
\end{proof}

In addition, Lemma~\ref{lem:ed} will be used to build the data structure $\calA_2$ in Lemma~\ref{lem:candidate}, and it will also help to solve the $k$-center problem in Section~\ref{sec:kcenter}. Its proof is given in Section~\ref{sec:ds}.

\begin{lemma}\label{lem:ed}
We can build a data structure $\calA_3$ in $O(M\log M)$ time
that can compute the expected distance $\Ed(x,P_i)$ in $O(\log M)$
time for any point $x\in T$ and any uncertain point $P_i\in \calP$.
\end{lemma}

\section{A Tree Decomposition and Computing the Medians}
\label{sec:decmedian}

In this section, we first introduce a decomposition of $T$, which will
be repeatedly used in our algorithms (e.g., for
Lemmas~\ref{lem:median}, \ref{lem:report}, \ref{lem:ed}).
Later in Section~\ref{sec:median} we will compute the medians with the help of the decomposition.

\subsection{A Connector-Bounded Centroid Decomposition}
\label{sec:decom}

We propose a tree decomposition of $T$, called a {\em connector-bounded centroid
decomposition}, which is different from the centroid decompositions used before,
e.g.,~\cite{ref:FredericksonFi83,ref:KarivAnC79,ref:MegiddoNe83,ref:MegiddoAn81}
and has certain properties that can facilitate our algorithms.

A vertex $v$ of $T$ is called a {\em centroid} if $T$ can be
represented as a union of two subtrees with $v$ as their only common
vertex and each subtree has at most $\frac{2}{3}$ of the vertices of
$T$~\cite{ref:KarivAnC79,ref:MegiddoAn81},
and we say the two subtrees are {\em decomposed} by $v$.
Such a centroid always exists and can be found in linear time~\cite{ref:KarivAnC79,ref:MegiddoAn81}.
For convenience of discussion, we consider $v$ to be contained in only one subtree
but an ``open vertex'' in the other subtree (thus, the location of $\calP$ at $v$ only belongs to one subtree).

Our decomposition of $T$ corresponds to a {\em decomposition tree}, denoted by $\Upsilon$ and defined recursively as follows. Each internal node of $\Upsilon$ has two, three, or four children.
The root of $\Upsilon$ corresponds to the entire tree $T$. Let $v$ be
a centroid of $T$, and let $T_1$ and $T_2$ be the subtrees of $T$
decomposed by $v$. Note that $T_1$ and $T_2$ are disjoint since we consider $v$ to be contained in only one of them.
Further, we call $v$ a {\em connector} in both $T_1$ and $T_2$.
Correspondingly, in $\Upsilon$, its root has two children
corresponding to $T_1$ and $T_2$, respectively.

In general, consider a node $\mu$ of $\Upsilon$. Let $T(\mu)$
represent the subtree of $T$ corresponding to $\mu$. We assume
$T(\mu)$ has at most two connectors (initially this is true when $\mu$
is the root). We further
decompose $T(\mu)$ into subtrees that correspond to the children of
$\mu$ in $\Upsilon$, as follows. Let $v$ be the centroid of $T(\mu)$
and let $T_1(\mu)$ and $T_2(\mu)$ respectively be the two subtrees of $T(\mu)$
decomposed by $v$. We consider $v$ as a {\em connector} in both
$T_1(\mu)$ and $T_2(\mu)$.

If $T(\mu)$ has at most one connector, then each of $T_1(\mu)$ and
$T_2(\mu)$ has at most two connectors. In this case, $\mu$ has two
children corresponding to $T_1(\mu)$ and $T_2(\mu)$, respectively.

%
%

If $T(\mu)$ has two connectors but each of  $T_1(\mu)$ and $T_2(\mu)$
still has at most two connectors (with $v$ as a new connector), then
$\mu$ has two children corresponding to $T_1(\mu)$ and $T_2(\mu)$,
respectively. Otherwise, one of them, say, $T_2(\mu)$, has three
connectors and the other $T_1(\mu)$ has only one connector (e.g., see Fig.~\ref{fig:decom}). In this case, $\mu$ has a child in $\Upsilon$
corresponding to $T_1(\mu)$, and we further
perform a {\em connector-reducing decomposition} on $T_2(\mu)$, as
follows (this is the main difference between our decomposition and the
traditional centroid decomposition used
before~\cite{ref:FredericksonFi83,ref:KarivAnC79,ref:MegiddoNe83,ref:MegiddoAn81}).
Depending on whether the three connectors of $T_2(\mu)$ are in a simple
path, there are two cases.

\begin{figure}[t]
\begin{minipage}[t]{\linewidth}
\begin{center}
\includegraphics[totalheight=1.5in]{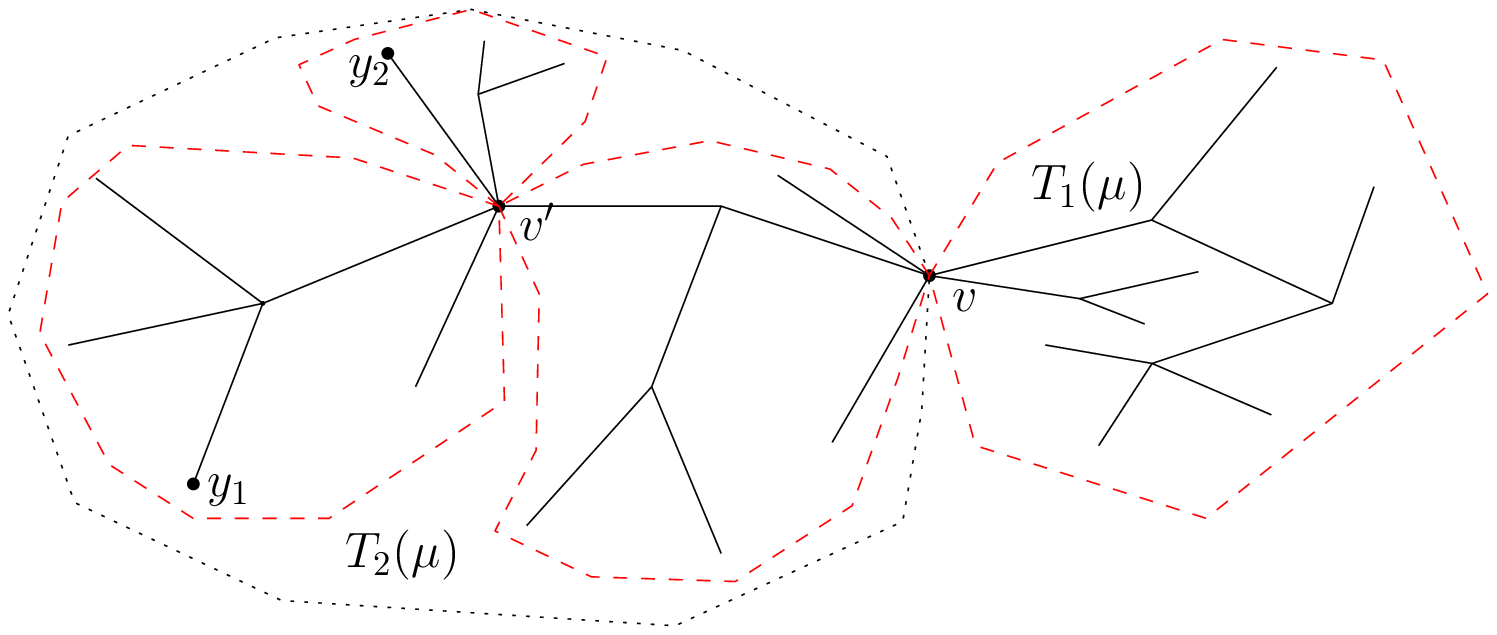}
\caption{\footnotesize Illustrating the decomposition of $T(\mu)$ into
four subtrees enclosed by the (red) dashed cycles, where $y_1$ and
$y_2$ are two connectors of $T(\mu)$. $T(\mu)$ is first
decomposed into two subtrees $T_1(\mu)$ and $T_2(\mu)$. However, since
$T_2(\mu)$ has three connectors, we further decompose it into three
subtrees each of which has at most two connectors. }
\label{fig:decom}
\end{center}
\end{minipage}
\vspace*{-0.15in}
\end{figure}

\begin{enumerate}
\item
If they are in a simple path, without loss of generality,
we assume $v$ is the one between the other two
connectors in the path. We decompose $T_2(\mu)$ into two subtrees at
$v$ such that they contain the two connectors respectively. In this
way, each subtree contains at most two connectors.
Correspondingly, $\mu$ has another two children corresponding the two
subtrees of $T_2(\mu)$, and thus $\mu$ has three children in total.

\item
Otherwise, there is a unique vertex $v'$ in $T_2(\mu)$ that
decomposes $T_2(\mu)$ into three subtrees that
contain the three connectors respectively (e.g., see
Fig.~\ref{fig:decom}). Note that $v'$ and the three
subtrees can be easily found in linear time by traversing $T_2(\mu)$.
Correspondingly, $\mu$ has another three children corresponding
to the above three subtrees of $T_2(\mu)$, respectively,
and thus $\mu$ has four children in total.
Note that we consider $v'$ as a connector in each of the above three
subtrees. Thus, each subtree contains at most two connectors.
\end{enumerate}

We continue the decomposition until each subtree $T(\mu)$ of $\mu\in
\Upsilon$ becomes an edge $e(v_1,v_2)$ of $T$.  According to our
decomposition, both $v_1$ and $v_2$ are connectors of $T(\mu)$, but they
may only open vertices of $T(\mu)$. If both $v_1$ and $v_2$ are open
vertices of  $T(\mu)$, then we will not further decompose $T(\mu)$, so
$\mu$ is a leaf of $\Upsilon$. Otherwise, we further decompose
$T(\mu)$ into an open edge and a closed vertex $v_i$ if $v_i$ is
contained in $T(\mu)$ for each $i=1,2$. Correspondingly, $\mu$  has
either two or three children that are leaves of $\Upsilon$.
In this way, for each leaf $\mu$ of $\Upsilon$, $T(\mu)$ is either an
open edge or a closed vertex of $T$. In the former case, $T(\mu)$ has
two connectors that are its incident vertices, and in the latter case,
$T(\mu)$ has one connector that is itself.

This finishes the decomposition of $T$. A major difference
between our decomposition and the traditional centroid
decomposition~\cite{ref:FredericksonFi83,ref:KarivAnC79,ref:MegiddoNe83,ref:MegiddoAn81}
is that the subtree in our decomposition has at most two connectors.
As will be clear later, this property is crucial to guarantee the
runtime of our algorithms.

\begin{lemma}
The height of $\Upsilon$ is $O(\log M)$ and $\Upsilon$ has $O(M)$ nodes.
The connector-bounded centroid decomposition of $T$ can be computed in $O(M\log M)$ time.
\end{lemma}
\begin{proof}
Consider any node $\mu$ of $\Upsilon$. Let $T(\mu)$ be the subtree of $T$
corresponding to $\mu$. According to our decomposition, $|T(\mu)|=O(M\cdot
(\frac{2}{3})^t)$, where $t$ is the depth of $\mu$ in $\Upsilon$. This
implies that the height of $\Upsilon$ is $O(\log M)$.

Since each leaf of $\Upsilon$ corresponds to either a vertex or an
open edge of $T$, the number of leaves of $\Upsilon$ is $O(M)$. Since
each internal node of $\Upsilon$ has at least two children, the number
of internal nodes is no more than the number of leaves. Hence,
$\Upsilon$ has $O(M)$ nodes.

According to our decomposition, all subtrees of $T$ corresponding to
all nodes in the same level of $\Upsilon$ (i.e., all nodes with the same
depth) are pairwise disjoint, and thus the total size of all these subtrees is $O(M)$.
Decomposing each subtree can be done in linear time (e.g., finding a
centroid takes linear time). Therefore, decomposing all subtrees in
each level of $\Upsilon$ takes $O(M)$ time. As the height of
$\Upsilon$ is $O(\log M)$, the total time for computing the
decomposition of $T$ is $O(M\log M)$. \qed
\end{proof}

In the following, we assume our decomposition of $T$ and the
decomposition tree $\Upsilon$ have been computed.
In addition, we introduce some notation that will be used later. For
each node $\mu$ of $\Upsilon$, we use $T(\mu)$ to represent the
subtree of $T$ corresponding to $\mu$. If $y$ is a connector of
$T(\mu)$, then we use $T(y,\mu)$ to represent the subtree of $T$
consisting of all points $q$ of $T\setminus T(\mu)$ such that $\pi(q,p)$ contains $y$
for any point $p\in T(\mu)$ (i.e., $T(y,\mu)$ is the ``outside world''
connecting to $T(\mu)$ through $y$; e.g., see Fig.~\ref{fig:median}).
By this definition, if $y$ is the only
connector of $T(\mu)$, then $T=T(\mu)\cup T(y,\mu)$; if $T(\mu)$ has
two connectors $y_1$ and $y_2$, then $T=T(\mu)\cup T(y_1,\mu)\cup
T(y_2,\mu)$.

\begin{figure}[t]
\begin{minipage}[t]{\linewidth}
\begin{center}
\includegraphics[totalheight=1.3in]{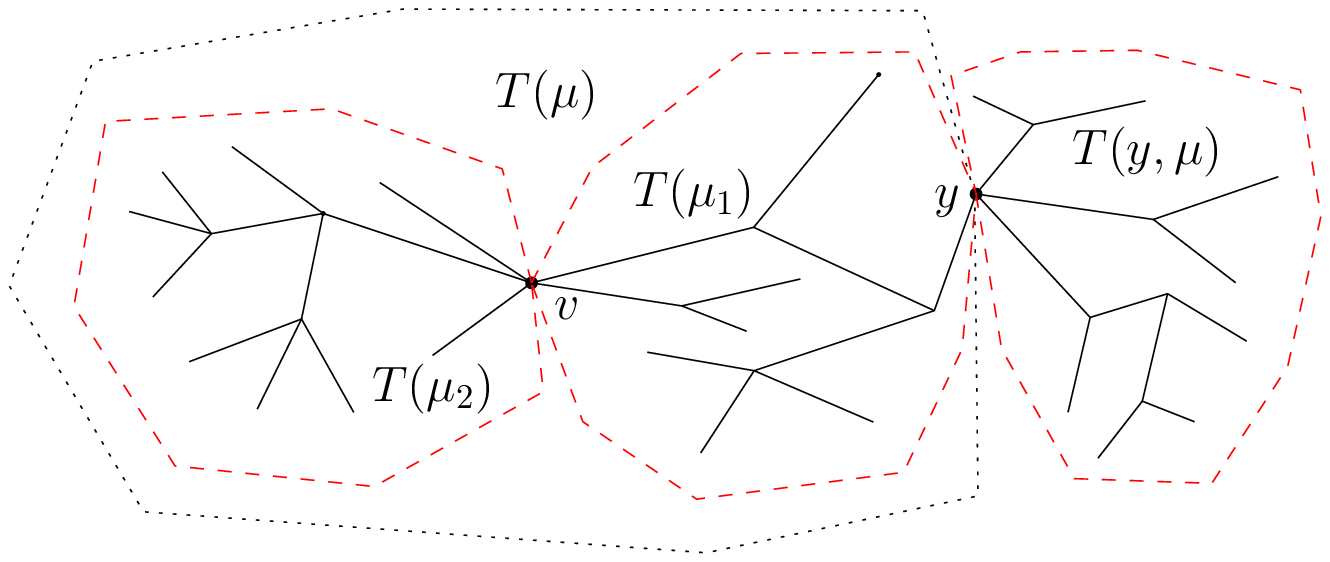}
\caption{\footnotesize Illustrating the subtrees $T(\mu_1),T(\mu_2)$,
and $T(y,\mu)$, where $y$ is a connector of $T(\mu)=T(\mu_1)\cup
T(\mu_2)$. Note that $T(y,\mu)$ is also $T(y,\mu_1)$ as $y\in T(\mu_1)$.}
\label{fig:median}
\end{center}
\end{minipage}
\vspace*{-0.15in}
\end{figure}

\subsection{Computing the Medians}
\label{sec:median}

In this section, we compute all medians.  It is easy to compute the median $p_i^*$ for
a single uncertain point $P_i$ in $O(M)$ time by traversing the tree
$T$. Hence, a straightforward algorithm can compute all $n$ medians in
$O(nM)$ time. Instead, we present an $O(M\log M)$ time algorithm,
which will prove  Lemma~\ref{lem:median}.


For any vertex $v$ (e.g., the centroid) of $T$, let $T_1$ and $T_2$ be
two subtrees of $T$ decomposed by $v$ (i.e., $v$ is
their only common vertex and $T=T_1\cup T_2$), such that
$v$ is contained in only one subtree and is an open
vertex in the other.
The following lemma can be readily obtained from Kariv and Hakimi
\cite{ref:KarivAn79}, and similar results were also given in \cite{ref:WangCo16}.

%

\begin{lemma}\label{lem:findmedian} 
For any uncertain point $P_i$ of $\calP$, we have the following.
\begin{enumerate}
\item
If the probability sum of $P_i$ in $T_j$ is greater than $0.5$ for
some $j\in \{1,2\}$, then the median $p_i^*$ must be in $T_j$.
\item
The vertex $v$ is $p_i^*$ if the probability sum of
$P_i$ in $T_j$ is equal to $0.5$ for some $j\in\{1,2\}$.
\end{enumerate}
\end{lemma}




Consider the connector-bounded centroid decomposition $\Upsilon$ of
$T$.
Starting from the root of $\Upsilon$,
our algorithm will process the nodes of $\Upsilon$ in a top-down
manner. Suppose we are processing a node $\mu$. Then, we maintain
a sorted list of indices for $\mu$, called the {\em index list} of
$\mu$ and denoted by $L(\mu)$, which
consists of all indices $i\in [1,n]$ such that $p_i^*$ is not found
yet but is known to be in the subtree $T(\mu)$. Since each index $i$
of $L(\mu)$ essentially refers to $P_i$, for convenience, we also say
that $L(\mu)$ is a set of uncertain points.
Let $F[1\cdots n]$ be an array, which will help to
compute the probability sums in our algorithm.

\subsubsection{The Root Case}
Initially, $\mu$ is the root and we process it as follows. We
present our algorithm in a way that is consistent with that for the general
case.

Since $\mu$ is the root, we have $T(\mu)=T$
and $L(\mu)=\{1,2,\ldots,n\}$. 
Let  $\mu_1$ and $\mu_2$ be the two children of $\mu$. Let $v$ be the
centroid of $T$ that is used to decompose $T(\mu)$ into $T(\mu_1)$ and
$T(\mu_2)$ (e.g., see Fig.~\ref{fig:median}).
We compute in $O(|T(\mu)|)$ time the probability sums of all uncertain
points of $L(\mu)$ in $T(\mu_1)$ by using the array $F$ and traversing $T(\mu_1)$.
Specifically, we first perform a {\em reset procedure} on $F$ to reset $F[i]$ to $0$ for
each $i\in L(\mu)$, by scanning the list $L(\mu)$.
Then, we traverse $T(\mu_1)$, and for each
visited vertex, which holds some uncertain point location $p_{ij}$, we
update $F[i]=F[i]+f_{ij}$. After the traversal, for each $i\in L(\mu)$, $F[i]$
is equal to the probability sum of $P_i$ in $T(\mu_1)$. By
Lemma~\ref{lem:findmedian}, if $F[i]=0.5$, then $p_i^*$ is $v$ and we
report $p_i^*=v$; if $F[i]>0.5$, then $p_i^*$ is in $T_1(\mu)$ and we
add $i$ to the end of the index list
$L(\mu_1)$ for $\mu_1$ (initially $L(\mu_1)=\emptyset$); if
$F[i]<0.5$, then $p_i^*$ is in $T_2(\mu)$ and add
$i$ to the end of $L(\mu_2)$ for $\mu_2$.
The above has correctly computed the index lists for $\mu_1$ and $\mu_2$.

Recall that $v$ is a connector in both $T(\mu_1)$ and $T(\mu_2)$. In
order to efficiently compute medians in $T(\mu_1)$ and $T(\mu_2)$
recursively, we compute a {\em probability list} $L(v,\mu_j)$ at $v$ for
$\mu_j$ for each $j=1,2$. We discuss $L(v,\mu_1)$ first.

The list $L(v,\mu_1)$ is the same as $L(\mu_1)$ except that each index
$i\in L(v,\mu_1)$ is also associated with a value, denoted by $F(i,v,\mu_1)$,
which is the probability sum of $P_i$ in $T(v,\mu_1)$ (recall the
definition of $T(v,\mu_1)$ at the end of Section~\ref{sec:decom}; note
that $T(v,\mu_1)=T(\mu_2)$ in this case).
The list $L(v,\mu_1)$ can be
built in $O(|T(\mu)|)$ time by traversing $T(\mu_2)$ and using the
array $F$. Specifically, we scan the list $L(\mu_1)$, and for each
index $i\in L(\mu_1)$, we reset $F[i]=0$.
Then, we traverse the subtree $T(\mu_2)$, and for each
location $p_{ij}$ in $T(\mu_2)$, we update $F[i]=F[i]+f_{ij}$
(if $i$ is not in $L(\mu_1)$, this step is actually redundant but does
not affect anything).
After the traversal, for each index
$i\in L(\mu_1)$, we copy it to $L(v,\mu_1)$ and set $F(i,v,\mu_1)=F[i]$.

Similarly, we compute the probability list $L(v,\mu_2)$ at $v$ for
$\mu_2$ in $O(|T(\mu)|)$ time by
traversing $T(\mu_1)$. This finishes the processing of the root $\mu$.
The total time is $O(|T(\mu)|)$ since $|L(\mu)|\leq |T(\mu)|$. Note that our
algorithm guarantees that for each $i\in L(\mu_1)$, $P_i$ must have at
least one location  in $T(\mu_1)$, and thus $|L(\mu_1)|\leq
|T(\mu_1)|$. Similarly, for each $i\in L(\mu_2)$, $P_i$ must have at
least one location  in $T(\mu_2)$, and thus $|L(\mu_2)|\leq
|T(\mu_2)|$.

\subsubsection{The General Case}
Let $\mu$ be an internal node of $\Upsilon$ such that the ancestors of $\mu$
have all been processed. Hence, we have a sorted index list $L(\mu)$. If
$L(\mu)=\emptyset$, then we do not need to process $\mu$ and any of
its descendants. We assume $L(\mu)\neq\emptyset$. Thus, for each $i\in
L(\mu)$, $p_i^*$ is in $T(\mu)$ and $P_i$ has at least one location in
$T(\mu)$ (and thus $|L(\mu)|\leq |T(\mu)|$).
Further, for each connector $y$ of $T(\mu)$, the algorithm maintains
a probability list $L(y,\mu)$ that is
the same as $L(\mu)$ except that each index $i\in L(y,\mu)$ is
associated with a value $F(i,y,\mu)$, which is the probability sum of $P_i$ in the
subtree $T(y,\mu)$.  Our processing algorithm for $\mu$ works as
follows, whose total time is $O(|T(\mu)|)$.

According to our decomposition, $T(\mu)$ has at
most two connectors and
$\mu$ may have two, three, or four children. We first
discuss the case where $\mu$ has two children,  and other cases can be handled similarly.

Let $\mu_1$ and $\mu_2$ be the two children of $\mu$, respectively.
Let $v$ be the centroid of $T(\mu)$ that is used to decompose it.
We discuss the subtree $T(\mu_1)$ first, and $T(\mu_2)$ is similar.
Since $v$ is a connector of $T(\mu_1)$ and
$T(\mu_1)$ has at most two connectors, $T(\mu_1)$ has at most one
connector $y$ other than $v$.  We consider the general situation where
$T(\mu_1)$ has such a connector $y$ (the case where such a connector
does not exist can be handled similarly but in a simpler way).
Note that $y$ must be a connector of $T(\mu)$.

We first compute the probability sums of $P_i$'s for all $i\in L(\mu)$ in the subtree
$T(\mu_1)\cup T(y,\mu)$ (e.g., see Fig.~\ref{fig:median}),
which can be done in $O(|T(\mu)|)$ time by
traversing $T(\mu_1)$ and using the array $F$ and
the probability list $L(y,\mu)$ at $y$, as follows. We scan the list $L(\mu)$
and for each index $i\in L(\mu)$, we reset $F[i]=0$. Then, we traverse
$T(\mu_1)$ and for each location $p_{ij}$, we update
$F[i]=F[i]+f_{ij}$ (it does not matter if $i\not\in L(\mu)$). When the traversal
visits $y$, we scan the list $L(y,\mu)$ and for each index $i\in L(y,\mu)$, we
update $F[i]=F[i]+F(i,y,\mu)$. After the traversal, for
each $i\in L(\mu)$, $F[i]$ is the probability sum of $P_i$ in
$T(\mu_1)\cup T(y,\mu)$. For each $i\in L(\mu)$, if $F[i]=0.5$,
we report $p_i^*=v$; if $F[i]>0.5$, we add $i$ to $L(\mu_1)$;
if $F[i]<0.5$, we add $i$ to $L(\mu_2)$. This builds the two lists
$L(\mu_1)$ and $L(\mu_2)$, which are initially $\emptyset$. Note that
since for each $i\in L(\mu)$, $P_i$ has at least one location in
$T(\mu)$, the above way of computing $L(\mu_1)$ (resp., $L(\mu_2)$)
guarantees that for each $i$ in $L(\mu_1)$ (resp., $L(\mu_2)$),
$P_i$ has at least one location in $T(\mu_1)$ (resp., $T(\mu_2)$), which implies $|L(\mu_1)|\leq |T(\mu_1)|$ (resp., $|L(\mu_2)|\leq |T(\mu_2)|$).

Next we compute the probability lists for the connectors of
$T(\mu_1)$. Note that $T(\mu_1)$ has two connectors $v$ and $y$.
For $v$, we compute the probability list $L(v,\mu_1)$
that is the same as $L(\mu_1)$ except that each $i\in L(v,\mu_1)$ is
associated with a value $F(i,v,\mu_1)$, which is the probability sum of
$P_i$ in the subtree $T(v,\mu_1)$.
To compute $L(v,\mu_1)$, we first reset $F[i]=0$ for each $i\in
L(\mu_1)$.  Then we traverse $T(\mu_2)$ and for each location
$p_{ij}\in T(\mu_2)$, we update $F[i]=F[i]+f_{ij}$. If $T(\mu_2)$ has
a connector $y'$ other than $v$, then $y'$ is also a connector of $T(\mu)$ (note that there
is at most one such connector);
we scan the probability list $L(y',\mu)$
and for each $i\in L(y',\mu)$, we update $F[i]=F[i]+F(i,y',\mu)$.
Finally, we scan $L(\mu_1)$ and for each $i\in L(\mu_1)$, we copy it to
$L(v,\mu_1)$ and set $F(i,v,\mu_1)=F[i]$. This computes the
probability list $L(v,\mu_1)$.

Further, we also need to compute the probability list
$L(y,\mu_1)$ at $y$ for $T(\mu_1)$. The list $L(y,\mu_1)$
is the same as $L(\mu_1)$ except that each
$i\in L(y,\mu_1)$ also has a value $F(i,y,\mu_1)$, which is the
probability sum of $P_i$ in $T(y,\mu_1)$.
To compute $L(y,\mu_1)$, we first copy all indices of $L(\mu_1)$ to
$L(y,\mu_1)$, and then compute the values $F(i,y,\mu_1)$, as follows.
Note that $T(y,\mu_1)$ is exactly $T(y,\mu)$ (e.g., see
Fig.~\ref{fig:median}). Recall that as a connector of $T(\mu)$, $y$
has a probability list $L(y,\mu)$ in which each $i\in L(y,\mu)$ has a
value $F(i,y,\mu)$. Notice that $L(y,\mu_1)\subseteq L(y,\mu)$.
Due to $T(y,\mu_1)=T(y,\mu)$,
for each $i\in L(y,\mu_1)$, $F(i,y,\mu_1)$ is equal to $F(i,y,\mu)$.
Since indices in each of  $L(y,\mu_1)$ and  $L(y,\mu)$ are
sorted, we scan $L(y,\mu_1)$ and $L(y,\mu)$ simultaneously (like merging two
sorted lists) and for each $i\in L(y,\mu_1)$, if we encounter $i$ in
$L(y,\mu)$, then we set $F(i,y,\mu_1)=F(i,y,\mu)$. This computes the
probability list $L(y,\mu_1)$ at $y$ for $T(\mu_1)$.

The above has processed the subtree $T(\mu_1)$.
Using the similar approach, we can process $T(\mu_2)$ and we omit the details.

This finishes the processing of $\mu$ for the case where $\mu$ has two
children. The total time is
$O(|T(\mu)|)$. To see this, the algorithm traverses $T(\mu)$ for a
constant number of times. The algorithm also visits the list $L(\mu)$ and the
probability list of each connector of $T(\mu)$ for a constant number
of times. Recall that $|L(\mu)|\leq |T(\mu)|$ and
$|L(\mu)|=|L(\mu,y)|$ for each connector $y$ of $T(\mu)$. Also recall
that $T(\mu)$ has at most two connectors. Thus, the total time
for processing $\mu$ is $O(|T(\mu)|)$.

\paragraph{Remark.} If the number of connectors of $T(\mu)$ were not
bounded by a constant, then we could not bound the processing time for
$\mu$ as above. This is one reason our decomposition on
$T$ requires each subtree $T(\mu)$ to have at most two connectors.
\vspace{0.1in}

If $\mu$ has three children, $\mu_1,\mu_2,\mu_3$, then
$T(\mu)$ is decomposed into three subtrees $T(\mu_j)$ for $j=1,2,3$. In this case, $T(\mu)$ has two connectors. To process $\mu$, we apply the above algorithm for
the two-children case twice.
Specifically, we consider the procedure of decomposing $T(\mu)$ into three subtrees consisting of two ``intermediate decomposition steps''.
According to our decomposition, $T(\mu)$ was first
decomposed into two subtrees by its centroid such that one subtree $T_1(\mu)$
contains at most two connectors while the other one $T_2(\mu)$ contains
three connectors, and we consider this as the first
intermediate step. The second intermediate step is to further decompose
$T_2(\mu)$ into two subtrees each of which contains at most two
connectors. To process $\mu$, we apply our two-children case algorithm on
the first intermediate step and then on the second intermediate step.
The total time is still $O(|T(\mu)|)$. We omit the details.

Similarly, if $\mu$ has four children, then the decomposition can be
considered as consisting
of three intermediate steps (e.g., in Fig.~\ref{fig:median}, the first
step is to decompose $T(\mu)$ into $T_1(\mu)$ and $T_2(\mu)$, and
then decomposing $T_2(\mu)$ into three subtrees can be considered as
consisting of two steps each of which decomposes a subtree into two
subtrees), and we apply our two-children case
algorithm three times. The total processing time for $\mu$ is also $O(|T(\mu)|)$.

The above describes the algorithm for processing $\mu$ when $\mu$ is
an internal node of $\Upsilon$.

If $\mu$ is a leaf, then $T(\mu)$ is either a vertex or an open
edge of $T$. If $T(\mu)$ is an open edge, the index list $L(\mu)$
must be empty since our algorithm only finds medians on
vertices. Otherwise, $T(\mu)$ is a vertex $v$ of $T$.
If $L(\mu)$ is not empty, then for each $i\in L(\mu)$, we simply
report $p_i^*=v$.

The running time of the entire algorithm is $O(M\log M)$. To see this,
processing each node $\mu$ of $\Upsilon$ takes $O(|T(\mu)|)$ time. For
each level of $\Upsilon$, the total sum of $|T(\mu)|$ of
all nodes $\mu$ in the level is $O(|T|)$. Since the height of
$\Upsilon$ is $O(\log M)$, the total time of the algorithm is $O(M\log
M)$.
This proves Lemma~\ref{lem:median}.

\section{The Data Structures $\calA_1$, $\calA_2$, and $\calA_3$}
\label{sec:ds}

In this section, we present the three data structures $\calA_1$,
$\calA_2$, and $\calA_3$, for Lemmas~\ref{lem:report}, \ref{lem:candidate}, and
\ref{lem:ed}, respectively. In particular, $\calA_3$ will be used to build
$\calA_2$ and it will also be needed for solving the $k$-center
problem in Section~\ref{sec:kcenter}.
Our connector-bounded centroid decomposition $\Upsilon$ will play an
important role in constructing both $\calA_1$ and $\calA_3$.
In the following, we present them in the order of
$\calA_1, \calA_3$, and $\calA_2$.

\subsection{The Data Structure $\calA_1$}
\label{sec:ds1}

The data structure $\calA_1$ is for answering the coverage-report-queries, i.e., given
any point $x\in T$, find all active uncertain points that are covered by $x$. Further, it
also supports the operation of removing an uncertain point once it is deactivated.

Consider any node $\mu\in \Upsilon$. If $\mu$ is the root, let
$L(\mu)=\emptyset$; otherwise, define $L(\mu)$ to be the sorted list
of all indices $i\in [1,n]$ such that $P_i$ does not have any
locations in the subtree $T(\mu)$ but has at least one location in
$T(\mu')$,
where $\mu'$ is the parent of $\mu$. Let $y$ be any connector of $T(\mu)$.
Let $L(y,\mu)$ be an index list the same as $L(\mu)$ and
each index $i\in L(y,\mu)$ is associated with two values: $F(i,y,\mu)$, which is the
probability sum of $P_i$ in the subtree $T(y,\mu)$,
and $D(i,y,\mu)$, which is the expected distance from $y$ to the locations of
$P_{i}$ in $T(y,\mu)$, i.e., $D(i,y,\mu)=w_i\cdot \sum_{p_{ij}\in T(y,\mu)}f_{ij}\cdot
d(p_{ij},y)$. We refer to $L(\mu)$ and $L(y,\mu)$ for each connector $y\in
T(\mu)$ as the {\em information lists} of $\mu$.

\begin{lemma}\label{lem:node}
Suppose $L(\mu)\neq \emptyset$ and the information lists of $\mu$ are
available. Let $t_{\mu}$ be the number of indices in $L(\mu)$. Then, we can build a
data structure of $O(t_{\mu})$ size in $O(|T(\mu)|+t_{\mu}\log t_{\mu})$ time
on $T(\mu)$, such that given any point $x\in T(\mu)$, we can report
all indices $i$ of $L(\mu)$ such that $P_i$ is covered by $x$ in
$O(\log n+ k\log n)$ amortized time, where $k$ is the output size; further, if
$P_i$ is deactivated with $i\in L(\mu)$, then we can remove $i$ from
the data structure and all information lists of $\mu$ in $O(\log n)$ amortized time.
\end{lemma}
\begin{proof}
As $L(\mu)\neq \emptyset$, $\mu$ is not the root. Thus, $T(\mu)$ has one or two connectors. We only discuss the most general case where $T(\mu)$ has
two connectors since the other case is similar but easier.
Let $y_1$ and $y_2$ denote the two connectors of $T(\mu)$, respectively. So the
lists $L(y_1,\mu)$ and $L(y_2,\mu)$ are available.

Note that for any two points $p$ and $q$ in $T(\mu)$, $\pi(p,q)$ is
also in $T(\mu)$ since $T(\mu)$ is connected.

\begin{figure}[t]
\begin{minipage}[t]{\linewidth}
\begin{center}
\includegraphics[totalheight=1.2in]{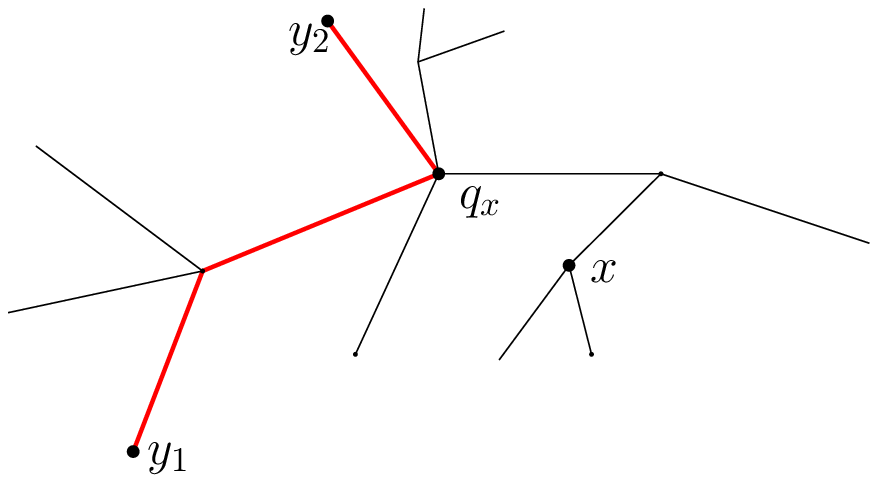}
\caption{\footnotesize Illustrating the definition of $q_x$ in the subtree $T(\mu)$ with two connectors $y_1$ and $y_2$. The path $\pi(y_1,y_2)$ is highlighted with thicker (red) segments.}
\label{fig:qx}
\end{center}
\end{minipage}
\vspace*{-0.15in}
\end{figure}

Consider any point $x\in T(\mu)$. Suppose we traverse on $T(\mu)$ from $x$ to
$y_1$, and let $q_x$ be the first point on $\pi(y_1,y_2)$ we encounter
(e.g. see Fig~\ref{fig:qx}; so $q_x$ is $x$ if $x\in\pi(y_1,y_2)$).
Let $a_x=d(x,q_x)$ and $b_x=d(q_x,y_1)$. Thus, $d(y_1,x)=a_x+b_x$ and
$d(y_2,x)=a_x+d(y_1,y_2)-b_x$.

For any $i\in L(\mu)$, since $P_i$ does not have any location in $T(\mu)$, we have
$F(i,y_1,\mu)+F(i,y_2,\mu)=1$, and thus the following holds for $\Ed(x,P_i)$:
\begin{equation*}
\begin{split}
&\Ed(x,P_i)  =  w_i\cdot \sum_{p_{ij}\in T}f_{ij}\cdot d(x,p_{ij}) \\
&= w_i\cdot \sum_{p_{ij}\in T(y_1,\mu)}f_{ij}\cdot d(x,p_{ij}) +
w_i\cdot \sum_{p_{ij}\in T(y_2,\mu)}f_{ij}\cdot d(x,p_{ij}) \\
&=w_i\cdot [F(i,y_1,\mu)\cdot (a_x+b_x) + D(i,y_1,\mu)] +
w_i\cdot [F(i,y_2,\mu)\cdot (a_x+d(y_1,y_2)-b_x) + D(i,y_2,\mu)] \\
&=w_i\cdot [a_x+(F(i,y_1,\mu)-F(i,y_2,\mu))\cdot
b_x+D(i,y_1,\mu)+D(i,y_2,\mu)+F(i,y_2,\mu)\cdot d(y_1,y_2)].
\end{split}
\end{equation*}

Notice that for any $x\in T(\mu)$, all above values are constant except $a_x$ and $b_x$.
Therefore, if we consider $a_x$ and $b_x$ as two variables of $x$,
$\Ed(x,P_i)$ is a linear function of them. In other words,
$\Ed(x,P_i)$ defines a plane in $\bbR^3$, where the $z$-coordinates
correspond to the values of $\Ed(x,P_i)$ and the $x$- and
$y$-coordinates correspond to $a_x$ and $b_x$ respectively. In the following, we also
use $\Ed(x,P_i)$ to refer to the plane defined by it in $\bbR^3$.

\paragraph{Remark.} This nice property for calculating $\Ed(x,P_i)$ is due to that
$\mu$ has at most two connectors. This is another reason
our decomposition requires every subtree $T(\mu)$ to have at most two
connectors.
\vspace{0.1in}

Recall that $x$ covers $P_i$ if $\Ed(x,P_i)\leq \lambda$. Consider the
plane $H_{\lambda}: z=\lambda$ in $\bbR^3$. In general the two planes $\Ed(x,P_i)$
and $H_{\lambda}$ intersect at a line $l_i$ and we let $h_i$ represent
the closed
half-plane of $H_{\lambda}$ bounded by $l_i$ and above the plane $\Ed(x,P_i)$. Let
$x_{\lambda}$ be the point $(a_x,b_x)$ in the plane $H_{\lambda}$. An easy
observation is that $\Ed(x,P_i)\leq \lambda$ if and only if $x_{\lambda}\in
h_i$. Further, we say that $l_i$ is an {\em upper bounding line} of $h_i$ if
$h_i$ is below $l_i$ and a {\em lower bounding line} otherwise.
Observe that if $l_i$ is an upper bounding line, then
$\Ed(x,P_i)\leq \lambda$ if and only if $x_{\lambda}$ is below $l_i$; if
$l_i$ is a lower bounding line, then $\Ed(x,P_i)\leq \lambda$ if and only
if $x_{\lambda}$ is above $l_i$.

Given any query point $x\in T(\mu)$, our goal for answering the query is to find all
indices $i\in L(\mu)$ such that $P_i$ is covered by $x$.  Based on the above
discussions, we do the following preprocessing. After $d(y_1,y_2)$ is
computed, by using the information lists of $y_1$ and $y_2$, we compute all functions
$\Ed(x,P_i)$ for all $i\in L(\mu)$ in $O(t_{\mu})$ time. Then, we
obtain a set $U$ of all upper bounding lines and a set of all
lower bounding lines on the plane $H_{\lambda}$ defined by
$\Ed(x,P_i)$ for all $i\in L(\mu)$. In the following, we
first discuss the upper bounding lines. Let $S_U$ denote the indices $i\in L(\mu)$ such that
$P_i$ defines an upper bounding line in $U$.

Given any point $x\in T(\mu)$, we first compute $a_x$ and $b_x$. This
can be done in constant time after $O(|T(\mu)|)$ time preprocessing,
as follows. In the preprocessing, for each vertex $v$ of $T(\mu)$,
we compute the vertex $q_v$ (defined in the similar way as $q_x$ with
respect to $x$) as well as the two values $a_v$ and $b_v$ (defined
similarly as $a_x$ and $b_x$, respectively). This can be easily done
in $O(|T(\mu)|)$ time by traversing $T(\mu)$ and we omit the details.
Given the point $x$, which is specified by an edge $e$ containing $x$,
let $v$ be the incident vertex of $e$ closer to $y_1$ and let $\delta$
be the length of $e$ between $v$ and $x$. Then, if $e$ is
on $\pi(y_1,y_2)$, we have $a_x=0$ and $b_x=b_v+\delta$. Otherwise,
$a_x=a_v+\delta$ and $b_x=b_v$.

After $a_x$ and $b_x$ are computed, the point $x_{\lambda}=(a_x,b_x)$ on the plane $H_{\lambda}$
is also obtained. Then, according to our discussion, all uncertain points
of $S_U$ that are covered by $x$ correspond to exactly those lines
of $U$ above $x_{\lambda}$. Finding the lines of $U$ above $x_{\lambda}$ is actually
the dual problem of half-plane range reporting query in $\bbR^2$.
By using the dynamic convex hull maintenance data
structure of Brodal and Jacob~\cite{ref:BrodalDy02}, with
$O(|U|\log |U|)$ time and $O(|U|)$ space preprocessing, for any point
$x_{\lambda}$, we can easily report all lines of $U$ above $x_{\lambda}$ in $O(\log |U| + k\log |U|)$ amortized
time (i.e., by repeating $k$ deletions),
where $k$ is the output size, and deleting a line from $U$ can be done in $O(\log |U|)$ amortized time. Clearly, $|U|\leq t_{\mu}$.

On the set of all lower bounding lines, we do the similar
preprocessing, and the query algorithm is symmetric.

Hence, the total preprocessing time is $O(|T(\mu)|+t_{\mu}\log
t_{\mu})$ time. Each query takes
$O(\log t_{\mu} + k\log^2 t_{\mu})$ amortized time and each remove operation can be performed in
$O(\log t_{\mu})$ amortized time.  Note that $t_{\mu}\leq n$. The lemma thus follows.
\qed
\end{proof}

The preprocessing algorithm for our data structure $\calA_1$ consists
of the following four steps. First, we compute the
information lists for all nodes $\mu$ of $\Upsilon$.
Second, for each node $\mu\in \Upsilon$, we compute the data structure
of Lemma~\ref{lem:node}. Third, for each $i\in [1,n]$, we compute
a {\em node list} $L_{\mu}(i)$ containing all nodes $\mu\in \Upsilon$
such that $i\in L(\mu)$. Fourth,
for each leaf $\mu$ of $\Upsilon$, if $T(\mu)$ is a vertex $v$ of $T$ holding
a location $p_{ij}$, then we maintain at $\mu$ the value $\Ed(v,P_i)$.
Before giving the details of the above processing algorithm,
we first assume the preprocessing work has been done and discuss the
algorithm for answering the coverage-report-queries.

Given any point $x\in T$, we answer the coverage-report-query as follows. Note that
$x$ is in $T(\mu_x)$ for some leaf $\mu_x$ of $\Upsilon$.
For each node $\mu$ in the path of $\Upsilon$ from the root  to
$\mu_x$, we apply the query algorithm
in Lemma~\ref{lem:node} to report all indices $i\in L(\mu)$ such that
$x$ covers $P_i$. In addition, if $T(\mu_x)$ is a vertex of $T$ holding a
location $p_{ij}$ such that $P_i$ is active, then we report $i$ if
$\Ed(v,P_i)$, which is maintained at $v$, is at most $\lambda$.
The following lemma proves the correctness and the
performance of our query algorithm.

\begin{lemma}
Our query algorithm correctly finds all active uncertain points
that are covered by $x$ in $O(\log M \log n+k\log n)$ amortized time, where $k$ is the
output size.
\end{lemma}
\begin{proof}
Let $\pi_x$ represent the path of $\Upsilon$ from the root to the leaf $\mu_x$.
To show the correctness of the algorithm, we argue that for each
active uncertain point $P_i$ that is covered by $x$, $i$ will be
reported by our query algorithm.

Indeed, if $T(\mu_x)$ is a vertex $v$ of $T$ holding a location $p_{ij}$ of
$P_i$, then the leaf $\mu_x$ maintains the value
$\Ed(v,P_i)$, which is equal to $\Ed(x,P_i)$ as $x=v$. Hence, our
algorithm will report $i$ when it processes $\mu_x$. Otherwise,
no location of $P_i$ is in $T(\mu_x)$. Since $P_i$ has locations in
$T$, if we go from the root to $\mu_x$ along $\pi$, we will eventually
meet a node $\mu$ such that $T(\mu)$ does not have any location of $P_i$
while $T(\mu')$ has at least one location of $P_i$, where $\mu'$ is
the parent of $\mu$. This implies that $i$ is in $L(\mu)$, and
consequently, our query algorithm will report $i$ when it processes
$\mu$. This establishes the correctness of our query algorithm.

For the runtime, as the height of $\Upsilon$ is $O(\log M)$, we make $O(\log M)$ calls on the query
algorithm in Lemma~\ref{lem:node}. Further, notice that each $i$ will be reported at most once. This is  because if $i$ is in $L(\mu)$ for some node $\mu$, then $i$ cannot be in $L(\mu')$ for any ancestor $\mu'$ of $\mu$. Therefore, the total runtime is $O(\log M \log n+k\log n)$.
\qed
\end{proof}

If an uncertain point $P_i$ is deactivated, then we scan the node list
$L_{\mu}(i)$ and for each node $\mu\in L_{\mu}(i)$, we remove $i$ from
the data structure by Lemma~\ref{lem:node}. The following
lemma implies that the total time is $O(m_i\log M\log n)$.

\begin{lemma}
For each $i\in [1,n]$, the number of nodes in $L_{\mu}(i)$ is
$O(m_i\log M)$.
\end{lemma}
\begin{proof}
Let $\alpha$ denote the number of nodes of $L_{\mu}(i)$.
Our goal is to argue that $i$ appears in $L(\mu)$ for
$O(m_i\log M)$ nodes $\mu$ of $\Upsilon$. Recall that if $i$ is in
$L(\mu)$ for a node $\mu\in \Upsilon$, then  $P_i$ has at least one
location in $T(\mu')$, where $\mu'$ is the parent of $\mu$. Since
each node of $\Upsilon$ has at most four children, if $N$ is the total
number of nodes $\mu'$ such that $P_i$ has at least one location in
$T(\mu')$, then
it holds that $\alpha\leq 4N$. Below we show that $N=O(m_i\log
M)$, which will prove the lemma.

Consider any location $p_{ij}$ of $P_i$. According to our
decomposition, the subtrees $T(\mu)$ for all nodes $\mu$ in the same
level of $\Upsilon$ are pairwise disjoint. Let $v$ be the vertex of $T$ that
holds $p_{ij}$, and let $\mu_v$ be the leaf of $\Upsilon$ with
$T(\mu_v)=v$. Observe that for any node $\mu\in \Upsilon$, $p_{ij}$
appears in $T(\mu)$ if and only if $\mu$ is in the path of $\Upsilon$ from
$\mu_v$ to the root. Hence, there are $O(\log M)$ nodes $\mu\in
\Upsilon$ such that $p_{ij}$ appears in $T(\mu)$.
As $P_i$ has $m_i$ locations, we obtain $N=O(m_i\log M)$.
\qed
\end{proof}

The following lemma gives our preprocessing algorithm for building $\calA_1$.

\begin{lemma}\label{lem:120}
$\sum_{\mu\in \Upsilon}t_{\mu}=O(M\log M)$, and the preprocessing time for constructing
the data structure $\calA_1$ excluding the second step is $O(M\log M)$.
\end{lemma}
\begin{proof}
We begin with the first step of the preprocessing algorithm for
$\calA_1$, i.e., computing the information lists for all nodes $\mu$ of $\Upsilon$.

In order to do so, for each node $\mu\in \Upsilon$, we will also compute a sorted
list $L'(\mu)$ of all such indices $i\in [1,n]$ that $P_i$ has at least one location
in $T(\mu)$, and further, for each connector $y$
of $T(\mu)$, we will compute a list $L'(y,\mu)$ that is
the same as $L'(\mu)$ except that each $i\in L'(y,\mu)$ is
associated with two values: $F(i,y,\mu)$, which is equal to the
probability sum of $P_i$ in the subtree $T(y,\mu)$, and $D(i,y,\mu)$,
which is equal to the expected distance from $y$ to the locations of $P_i$ in
$T(y,\mu)$, i.e., $D(i,y,\mu)=w_i\cdot \sum_{p_{ij}\in
T(y,\mu)}f_{ij}\cdot d(y,p_{ij})$. With a little abuse of notation,
we call all above the {\em information lists} of $\mu$ (including its
original information lists). In the
following, we describe our algorithm for computing the information
lists of all nodes $\mu$ of $\Upsilon$. Let
$F[1\cdots n]$ and $D[1\cdots n]$ be two arrays that we are going to
use in our algorithm (they will mostly be used to compute the $F$
values and $D$ values of the information lists of connectors).

Initially, if $\mu$ is the root of $\Upsilon$, we have
$L'(\mu)=\{1,2,\ldots,n\}$ and $L(\mu)=\emptyset$. Since $T(\mu)$ does
not have any connectors, we do not need to compute the information
lists for connectors.

Consider any internal node $\mu$. We assume all information lists for $\mu$ has
been computed (i.e., $L(\mu)$, $L'(\mu)$, and $L'(y,\mu)$, $L(y,\mu)$ for each
connector $y$ of $T(\mu)$). In the following we present our algorithm for
processing $\mu$, which will compute
the information lists of all children of $\mu$ in $O(|T(\mu)|)$ time.

We first discuss the case where $\mu$ has two children, denoted by
$\mu_1$ and $\mu_2$, respectively. Let $v$ be the centroid of $T(\mu)$
that is used to decompose $T(\mu)$ into $T(\mu_1)$ and $T(\mu_2)$
(e.g., see Fig.~\ref{fig:inforlist}).
We first compute the information lists of $\mu_1$, as follows.

We begin with computing the two lists $L(\mu_1)$ and $L'(\mu_1)$. Initially, we set
both of them to $\emptyset$. We scan the list $L'(\mu)$ and for each
$i\in L'(\mu)$, we reset
$F[i]=0$. Then, we scan the subtree $T(\mu_1)$, and for each location
$p_{ij}$, we set $F[i]=1$ as a flag showing that $P_i$ has locations in
$T(\mu_1)$. Afterwards, we scan the list $L'(\mu)$ again, and for
each $i\in L'(\mu)$, if $F[i]=1$, then we add $i$ to $L'(\mu_1)$;
otherwise, we add $i$ to $L(\mu_1)$. This computes the two index lists
$L(\mu_1)$ and $L'(\mu_1)$ for $\mu_1$. The running time is $O(|T(\mu)|)$
since the size of $L'(\mu)$ is no more than $|T(\mu)|$.

We proceed to compute the information lists for the connectors of
$T(\mu_1)$. Recall that $v$ is a connector of $T(\mu_1)$. So we need to
compute the two lists $L(v,\mu_1)$ and $L'(v,\mu_1)$, such that each
index $i$ in either list is associated with the two values
$F(i,v,\mu_1)$ and $D(i,v,\mu_1)$. We first copy all indices of
$L(\mu_1)$ to $L(v,\mu_1)$ and copy all indices of $L'(\mu_1)$ to
$L'(v,\mu_1)$. Next we compute their $F$ and $D$ values as follows.

\begin{figure}[t]
\begin{minipage}[t]{\linewidth}
\begin{center}
\includegraphics[totalheight=1.3in]{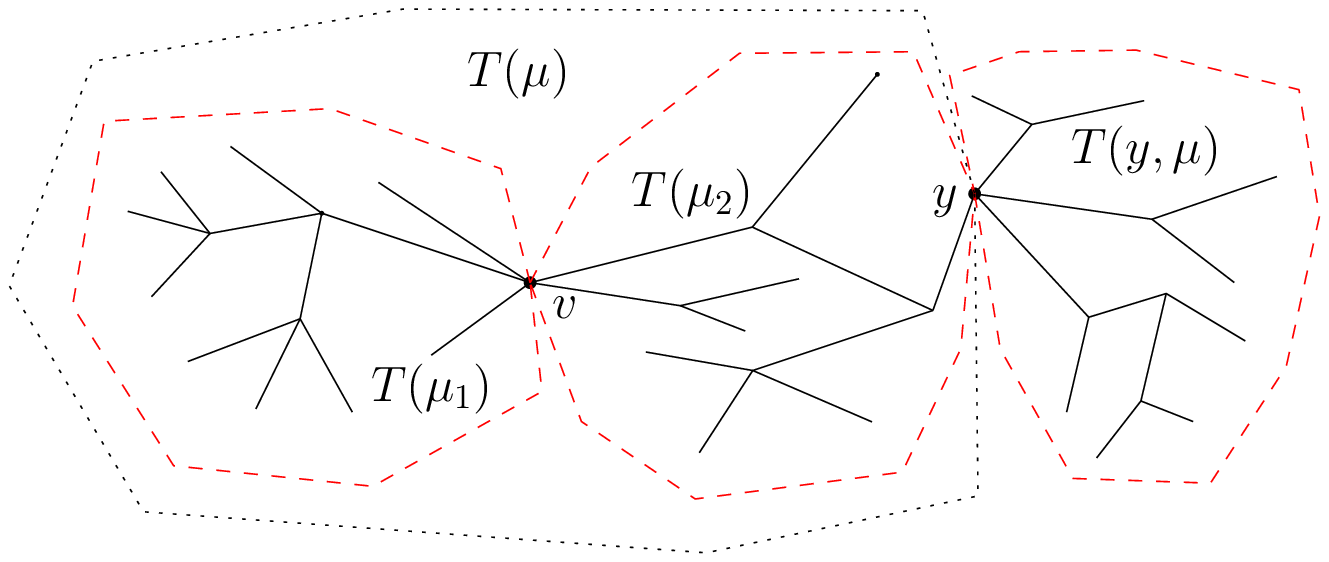}
\caption{\footnotesize Illustrating the subtrees $T(\mu_1),T(\mu_2)$,
and $T(y,\mu)$, where $y$ is a connector of $T(\mu)=T(\mu_1)\cup
T(\mu_2)$. Note that $T(y,\mu)$ is also $T(y,\mu_2)$ as $y\in T(\mu_2)$.}
\label{fig:inforlist}
\end{center}
\end{minipage}
\vspace*{-0.15in}
\end{figure}

We first scan $L'(\mu)$ and for each $i\in L'(\mu)$,
we reset $F[i]=0$ and $D[i]=0$. Next, we traverse $T(\mu_2)$ and for
each location $p_{ij}$, we update $F[i]=F[i]+f_{ij}$ and
$D[i]=D[i]+w_i\cdot f_{ij}\cdot d(v,p_{ij})$ ($d(v,p_{ij})$ can be computed in constant
time after $O(T(\mu))$-time preprocessing that computes $d(v,v')$ for
every vertex $v'\in T(\mu)$ by traversing $T(\mu)$). Further, if
$T(\mu_2)$ has a connector $y$ other than $v$, then
$y$ must be a connector of $T(\mu)$ (e.g., see
Fig.~\ref{fig:inforlist}; there exists
at most one such connector $y$); we scan the list $L'(y,\mu_2)$,
and for each $i\in L'(y,\mu_2)$, we update $F[i]=F[i]+F(i,y,\mu)$ and
$D[i]=D[i]+D(i,y,\mu)+w_i\cdot d(v,y)\cdot F(i,y,\mu)$ ($d(v,y)$ is already
computed in the preprocessing discussed above).
Finally, we scan $L(v,\mu_1)$ (resp., $L'(v,\mu_1)$) and for each
index $i$ in $L(v,\mu_1)$ (resp., $L'(v,\mu_1)$), we set
$F(i,v,\mu_1)=F[i]$ and $D(i,v,\mu_1)=D[i]$. This computes the two
information lists $L(v,\mu_1)$ and $L'(v,\mu_1)$. The total time is $O(|T(\mu)|)$.

In addition, if $T(\mu_1)$ has a connector $y$ other than $v$, then
$y$ must be a connector of $T(\mu)$ (e.g., see Fig.~\ref{fig:median}; there is only one such
connector), and we
further compute the two information lists $L(y,\mu_1)$ and
$L'(y,\mu_1)$. To do so, we first copy all indices of $L(\mu_1)$ to $L(y,\mu_1)$
and copy all indices of $L'(\mu_1)$ to $L'(y,\mu_1)$. Observe that
$L(y,\mu_1)$ and $L'(y,\mu_1)$ form a partition of the
indices of $L'(y,\mu)$. For each index $i$ in $L(y,\mu_1)$ (resp.,
$L'(y,\mu_1)$), we have $F(i,y,\mu_1)=F(i,y,\mu)$ and
$D(i,y,\mu_1)=D(i,y,\mu)$. Therefore, the $F$ and $D$ values for
$L'(y,\mu_1)$ and $L(y,\mu_1)$ can be obtained from $L'(y,\mu)$ by scanning the three
lists $L'(y,\mu_1)$, $L(y,\mu_1)$, and $L'(y,\mu)$ simultaneously, as
they are all sorted lists.

The above has computed the information lists for $\mu_1$ and the total time is $O(|T(\mu)|)$. Using the similar approach, we can compute the information lists for
$\mu_2$, and we omit the details. This finishes the algorithm for
processing $\mu$ where $\mu$ has two children.

If $\mu$ has three children, then $T(\mu)$ is decomposed into three
subtrees in our decomposition. As discussed in
Section~\ref{sec:median} on the algorithm for
Lemma~\ref{lem:median}, we can consider the decomposition of $T(\mu)$
consisting of two intermediate decomposition steps each of which decompose a subtree into two subtrees.
For each intermediate step, we
apply the above processing algorithm for the two-children case. In this way, we can compute the information lists for all three children of $\mu$ in $O(|T(\mu)|)$ time.
If $\mu$ has four children, then similarly there are four intermediate
decomposition steps and we apply the two-children case
algorithm three times. The total processing time for $\mu$ is
still $O(|T(\mu)|)$.

Once all internal nodes of $\Upsilon$ are processed, the
information lists of all nodes are computed. Since processing
each node $\mu$ of $\Upsilon$ takes $O(|T(\mu)|)$ time, the total
time of the algorithm is $O(M\log M)$. This also implies that the
total size of the information lists of all nodes of $\Upsilon$ is $O(M\log M)$,
i.e., $\sum_{\mu\in \Upsilon}t_{\mu}=O(M\log M)$.

This above describes the first step of our preprocessing algorithm for $\calA_1$.
For the third step, the node lists $L_{\mu}(i)$ can be built during the course of the
above algorithm. Specifically, whenever an index $i$ is added to $L(\mu)$ for some
node $\mu$ of $\Upsilon$, we add $\mu$ to the list $L_{\mu}(i)$. This
only introduces constant extra time each. Therefore, the overall algorithm has
the same runtime asymptotically as before.

For the fourth step, for each leaf $\mu$ of $\Upsilon$ such that $T(\mu)$ is a vertex $v$
of $T$, we do the following. Let $p_{ij}$ be the uncertain point location at $v$.
Based on our above algorithm, we have $L'(\mu)=\{i\}$. Since
$v$ is a connector, we have a list $L'(v,\mu)$ consisting of $i$
itself and two values $F(i,v,\mu)$ and
$D(i,v,\mu)$. Notice that $\Ed(v,P_i)=D(i,v,\mu)$. Hence, once the
above algorithm finishes, the value $\Ed(v,P_i)$ is available.

As a summary, the preprocessing algorithm for $\calA_1$ except the second step
runs in  $O(M\log M)$ time.  The lemma thus follows.
\qed
\end{proof}

For the second step of the preprocessing of $\calA_1$,
since $\sum_{\mu\in \Upsilon}t_{\mu}=O(M\log M)$ by
Lemma~\ref{lem:120}, applying the
preprocessing algorithm of Lemma~\ref{lem:node} on all nodes of
$\Upsilon$ takes $O(M\log^2 M)$ time and $O(M\log M)$ space in total.
Hence, the total preprocessing time of $\calA_1$ is
$O(M\log^2 M)$ and the space is $O(M\log M)$.
This proves Lemma~\ref{lem:report}.

\subsection{The Data Structure $\calA_3$}
\label{sec:ed}

In this section, we present the data structure $\calA_3$. Given any point
$x$ and any uncertain point $P_i$, $\calA_3$ is used to compute the
expected distance $\Ed(x,P_i)$.
Note that we do not need to consider the remove operations for $\calA_3$.

We follow the notation defined in Section~\ref{sec:ds1}. As
preprocessing, for each node $\mu\in \Upsilon$,
we compute the information lists $L(\mu)$ and $L(y,\mu)$ for each connector $y$ of
$T(\mu)$. This is actually the first step of the preprocessing
algorithm of $\calA_1$ in Section~\ref{sec:ds1}.
Further, we also preform the fourth step of the preprocessing
algorithm for $\calA_1$.
The above can be done in $O(M\log M)$ time by Lemma~\ref{lem:120}.

Consider any node $\mu\in \Upsilon$ with $L(\mu)\neq\emptyset$. Given
any point $x\in T(\mu)$, we have shown in the proof of Lemma~\ref{lem:node} that
$\Ed(x,P_i)$ is a function of two variables $a_x$ and $b_x$. As preprocessing, we compute
these functions for all $i\in L(\mu)$, which takes $O(t_{\mu})$
time as shown in the proof of Lemma~\ref{lem:node}.
For each $i\in L(\mu)$, we store the function $\Ed(x,P_i)$ at
$\mu$. 
The total preprocessing time for $\calA_3$ is $O(M\log M)$.

Consider any query on a point $x\in T$ and $P_i\in \calP$.
Note that $x$ is specified by an edge $e$ and its distance to
a vertex of $e$. Let $\mu_x$ be the leaf of $\Upsilon$ with $x\in
T(\mu_x)$. If $x$ is in the interior of $e$, then $T(\mu_x)$ is the
open edge $e$; otherwise, $T(\mu_x)$ is a single vertex $v=x$.

We first consider the case where $x$ is in the interior of $e$.
In this case, $P_i$ does not have any location in $T(\mu_x)$ since
$T(\mu_x)$ is an open edge.
Hence, if we go along the path of $\Upsilon$ from the root to $\mu_x$, we will
encounter a first node $\mu'$ with $i\in L(\mu')$. After
finding $\mu'$, we compute $a_x$ and $b_x$ in $T(\mu')$, which can
be done in constant time after $O(|T(\mu')|)$ time preprocessing on
$T(\mu')$, as discussed in the proof of Lemma~\ref{lem:node} (so the
total preprocessing time for all nodes of $\Upsilon$ is $O(M\log M)$).
After $a_x$ and $b_x$ are computed,
we can obtain the value $\Ed(x,P_i)$.

\paragraph{Remark.}
One can verify (from the proof of
Lemma~\ref{lem:node}) that as $x$ changes on $e$, $\Ed(x,P_i)$
is a linear function of $x$ because one of $a_x$ and $b_x$ is constant and the other linearly changes as $x$ changes in $e$. Hence, the above also computes the linear function $\Ed(x,P_i)$ for $x\in e$.
\vspace{0.07in}

To find the above node $\mu'$, for each node $\mu$ in the path of
$\Upsilon$ from
the root to $\mu_x$, we need to determine whether $i\in L(\mu)$. If we
represented the sorted index list $L(\mu)$ by a binary search tree, then
we could spend $O(\log n)$ time on each node $\mu$ and thus the total query
time would be $O(\log n\log M)$.
To remove the $O(\log n)$ factor, we further enhance our
preprocessing work by building a fractional cascading structure
\cite{ref:ChazelleFr86} on the
sorted index lists $L(\mu)$ for all nodes $\mu$ of $\Upsilon$. The
total preprocessing time for building the
structure is linear in the total number of nodes of all lists, which
is $O(M\log M)$ by Lemma~\ref{lem:120}.
For each node $\mu$, the fractional cascading structure will
create a new list $L^*(\mu)$ such that $L(\mu)\subseteq L^*(\mu)$.
Further, for each index $i\in L^*(\mu)$, if it is also in $L(\mu)$,
then we set a flag as an indicator. Setting the flags for all nodes of
$\Upsilon$ can be done in $O(M\log M)$ time as well.
Using the fractional cascading structure, we only
need to do binary search on the list in the root and then spend
constant time on each subsequent node \cite{ref:ChazelleFr86}, and thus the total query time
is $O(\log M)$.

If $x$ is a vertex $v$ of $T$, then depending on whether the
location at $v$ is $P_i$'s or not, there are two subcases. If it is
not, then we apply the same query algorithm as above. Otherwise,
let $p_{ij}$ be the location at $v$. Recall that in our preprocessing,
the value $\Ed(v,P_i)$ has already been computed and stored at
$\mu_x$ as $T(\mu_x)=v$. Due to $v=x$, we obtain $\Ed(x,P_i)=\Ed(v,P_i)$.

Hence, in either case, the query algorithm runs in $O(\log M)$ time.
This proves Lemma~\ref{lem:ed}.

\subsection{The Data Structure $\calA_2$}

The data structure $\calA_2$ is for answering candidate-center-queries: Given any
vertex $v\in T_m$, the query asks for the candidate center $c$ for the
active medians in $T_m(v)$, which is the subtree of $T_m$ rooted at
$v$. Once an uncertain point is deactivated, $\calA_2$ can also
support the operation of removing it.

Consider any vertex $v\in T_m$. Recall that due to our reindexing, the indices of all
medians in $T_m(v)$ exactly form the range $R(v)$. Recall that the candidate center
$c$ is the point on the path $\pi(v,r)$ closest to $r$ with $\Ed(v,P_i)\leq
\lambda$ for each active uncertain point $P_i$ with $i\in R(v)$.
Also recall that our algorithm invariant guarantees that whenever a
candidate-center-query is called at a vertex $v$, then it holds that
$\Ed(v,P_i)\leq\lambda$ for each active uncertain point $P_i$ with
$i\in R(v)$. However, we actually give a result that can answer a more
general query. Specifically, given a range $[k,j]$ with $1\leq k\leq
j\leq n$, let $v_{kj}$ be the lowest common ancestor of all medians $p_i^*$ with $i\in [k,j]$ in $T_m$; if $\Ed(v_{kj},P_i)>\lambda$ for some active $P_i$ with
$i\in [k,j]$, then our query algorithm will return $\emptyset$; otherwise,
our algorithm will compute a point $c$ on $\pi(v_{kj},r)$ closest to $r$ with $\Ed(c,P_i)\leq \lambda$ for each active $P_i$ with $i\in [k,j]$. We refer to it as
the {\em generalized} candidate-center-query.

In the preprocessing, we build a complete binary search tree $\calT$ whose
leaves from left to right correspond to indices $1,2,\ldots,n$.
For each node $u$ of $\calT$, let $R(u)$ denote the set of indices
corresponding to the leaves in the subtree of $\calT$ rooted at $u$.
For each median $p_i^*$, define $q_i$ to be the point $x$ on the
path $\pi(p_i^*,r)$ of $T_m$ closest to $r$ with $\Ed(x,P_i)\leq
\lambda$.

For each node $u$ of $\calT$, we define a node $q(u)$ as
follows. If $u$ is a leaf, define
$q(u)$ to be $q_i$, where $i$ is the index corresponding to leaf $u$.
If $u$ is an internal node, let $v_u$ denote the vertex of $T_m$ that
is the lowest common ancestor of the medians $p_i^*$ for all $i\in
R(u)$. If $\Ed(v_u,P_i)\leq \lambda$ for all $i\in R(u)$ (or
equivalently, $q_i$ is in $\pi(v_u,r)$ for all $i\in R(u)$), then
define $q(u)$ to be the point $x$ on the
path $\pi(v_u,r)$ of $T_m$ closest to $r$ with $\Ed(x,P_i)\leq
\lambda$ for all $i\in R(u)$; otherwise, $q(u)=\emptyset$.

\begin{lemma}\label{lem:110}
The points $q(u)$ for all nodes $u\in \calT$ can be computed in
$O(M\log M + n\log^2 M)$ time.
\end{lemma}
\begin{proof}
Assume the data structure $\calA_3$ for Lemma~\ref{lem:ed} has been
computed in $O(M\log M)$ time. In the following, by using $\calA_3$
we compute $q(u)$ for all nodes $u\in \calT$ in $O(M+n\log^2 M)$ time.

We first compute $q_i$ for all medians $p_i^*$.
Consider the depth-first-search on $T_m$ starting from the
root $r$. During the traversal, we use a
stack $S$ to maintain all vertices in order along the path $\pi(r,v)$
whenever a vertex $v$ is visited. Such a stack can be easily
maintained by standard techniques (i.e., push new vertices into $S$
when we go ``deeper'' and pop vertices out of $S$ when backtrack),
without affecting the linear-time performance of the traversal
asymptotically. Suppose the traversal visits a median $p_i^*$. Then,
the vertices of $S$ essentially form the path $\pi(r,p_i^*)$.
To compute $q_i$, we do binary search on the vertices of
$S$, as follows.

We implement $S$ by using an array of size $M$. Since
the order of the vertices of $S$ is the same as their order along
$\pi(r,p_i^*)$, the expected distances $\Ed(v,P_i)$ of the vertices
$v\in S$ along their order in $S$ are monotonically changing.
Consider a middle vertex $v$ of $S$. The vertex $v$ partitions $S$ into
two subarrays such that one subarray contains all vertices of
$\pi(r,v)$ and the other contains vertices of $\pi(v,p_i^*)$.  We
compute $\Ed(v,P_i)$ by using data structure $\calA_3$. Depending on
whether $\Ed(v,P_i)\leq \lambda$, we can proceed on only one
subarray of $M$. The binary search will eventually locate an edge
$e=(v,v')$ such that $\Ed(v,P_i)\leq \lambda$ and $\Ed(v',P_i)>
\lambda$. Then, we know that $q_i$ is located on $e\setminus \{v'\}$.
We further pick any point $x$ in the interior of $e$ and the data
structure $\calA_3$ can also compute the function $\Ed(x,P_i)$ for
$x\in e$ as remarked in Section~\ref{sec:ed}. With the function $\Ed(x,P_i)$ for $x\in e$, we can
compute $q_i$ in constant time. Since the binary search calls $\calA_3$
$O(\log M)$ times, the total time of the binary search is $O(\log^2 M)$.

In this way, we can compute $q_i$ for all medians $p_i^*$ with $i\in
[1,n]$ in $O(M+n\log^2 M)$ time, where the $O(n\log^2 M)$ time is for the binary
search procedures in the entire algorithm and the $O(M)$ time is for
traversing the tree $T_m$. Note that this also computes
$q(u)$ for all leaves $u$ of $\calT$.

We proceed to compute the points $q(u)$ for all internal nodes $\mu$ of
$\calT$ in a bottom-up manner.  Consider an
internal node $u$ such that $q(u_1)$ and $q(u_2)$ have been computed,
where $u_1$ and $u_2$ are the children of $u$, respectively. We
compute $q(u)$ as follows.

If either one of $q(u_1)$ and $q(u_2)$ is $\emptyset$, then we set
$q(u)=\emptyset$.  Otherwise, we do the following. Let
$i$ (resp., $j$) be the leftmost (resp., rightmost) leaf in the
subtree $\calT(u)$ of $\calT$ rooted at $u$. We first find the lowest comment ancestor of
$p_i^*$ and $p_j^*$ in the tree $T_m$, denoted by $v_{ij}$.
Due to our particular way of defining indices of all medians, $v_{ij}$ is the
lowest common ancestor of the medians $p_k^*$ for all $k\in [i,j]$. We determine
whether $q(u_1)$ and $q(u_2)$ are both on $\pi(r,v_{ij})$. If either
one is not on  $\pi(r,v_{ij})$, then we set $q(u)=\emptyset$;
otherwise, we set $q(u)$ to the one of $q(u_1)$ and $q(u_2)$
closer to $v_{ij}$.

The above for computing $q(u)$ can be implemented in $O(1)$ time,
after $O(M)$ time preprocessing on $T_m$. Specifically, with $O(M)$ time preprocessing on $T_m$, given any two vertices of $T_{m}$, we can compute their lowest
common ancestor in $O(1)$ time~\cite{ref:BenderTh00,ref:HarelFa84}.
Hence, we can compute $v_{ij}$ in constant time. To determine
whether $q(u_1)$ is on $\pi(r,v_{ij})$, we use the following approach. As a point on $T_m$, $q(u_1)$ is specified by an edge $e_1$ and its distance to one incident vertex of $e_1$. Let $v_1$ be the incident vertex of $e_1$ that is farther from the root $r$.
Observe that $q(u_1)$ is on
$\pi(r,v_{ij})$ if and only if the lowest common ancestor of $v_1$
and $v_{ij}$ is $v_1$. Hence, we can determine whether
$q(u_1)$ is on $\pi(r,v_{ij})$ in constant time by a lowest common ancestor query. Similarly, we can
determine whether $q(u_2)$ is on $\pi(r,v_{ij})$ in constant time.
Assume both $q(u_1)$ and $q(u_2)$ are on $\pi(r,v_{ij})$. To determine
which one of $q(u_1)$ and $q(u_2)$ is closer to $v_{ij}$, if they are on the same edge $e$ of $T_m$, then this can be done in constant time since both points are specified by their distances to an incident vertex of $e$. Otherwise, let $e_1$ be the edge of $T_m$ containing $q(u_1)$ and let $v_1$ be the incident vertex of $e_1$ farther to $r$; similarly, let $e_2$ be the edge of $T_m$ containing $q(u_2)$ and let $v_2$ be the incident vertex of $e_2$ farther to $r$. Observe that
$q(u_1)$ is closer to $v_{ij}$ if and only if the lowest common
ancestor of $v_1$ and $v_2$ is $v_2$, which can be
determined in constant time by a lowest common ancestor query.

The above shows that we can compute $q(u)$ in constant time based on
$q(u_1)$ and $q(u_2)$. Thus, we can compute $q(u)$ for all internal
nodes $u$ of $\calT$ in $O(n)$ time. The lemma thus follows.
\qed
\end{proof}

In addition to constructing the tree $\calT$ as above, our
preprocessing for $\calA_2$ also includes building
a lowest common ancestor query data structure on $T_m$ in
$O(M)$ time, such that given any two vertices of $T_{m}$, we can compute their lowest
common ancestor in $O(1)$ time~\cite{ref:BenderTh00,ref:HarelFa84}.
This finishes the preprocessing for $\calA_2$. The total time is
$O(M\log M+n\log^2 M)$.

The following
lemma gives our algorithm for performing operations on $\calT$.

\begin{lemma}
Given any range $[k,j]$, we can answer each generalized
candidate-center-query in $O(\log n)$ time, and each remove operation (i.e.,
deactivating an uncertain point) can be performed in $O(\log n)$ time.
\end{lemma}
\begin{proof}
We first describe how to perform the remove operations. Suppose an
uncertain point $P_i$ is deactivated. Let $u_i$ be the leaf of $\calT$
corresponding to the index $i$. We first set $q(u_i)=\emptyset$.
Then, we consider the path of $\calT$ from $u_i$ to the root in a bottom-up manner,
and for each node $u$, we update
$q(u)$ based on $q(u_1)$ and $q(u_2)$ in constant time in exactly the same way as in
Lemma~\ref{lem:110}, where $u_1$ and $u_2$ are the two children of
$u$, respectively. In this way, each remove operation can be
performed in $O(\log n)$ time.

Next we discuss the generalized candidate-center-query on a range $[k,j]$. By
standard techniques, we can locate a set $S$ of $O(\log n)$ nodes of $\calT$ such
that the descendant leaves of these nodes exactly correspond to
indices in the
range $[k,j]$. We find the lowest common ancestor $v_{kj}$ of $p_k^*$
and $p_j^*$ in $T_{m}$ in constant time. Then, for
each node $u\in S$, we check whether $q(u)$ is on $\pi(r,v_{kj})$,
which can be done in constant time by using the lowest common ancestor
query in the same way as in the proof of Lemma~\ref{lem:110}.
If $q(u)$ is not on $\pi(r,v_{kj})$ for some $u\in S$, then we simply
return $\emptyset$.
Otherwise, $q(u)$ is on $\pi(r,v_{kj})$ for every $u\in S$. We further find
the point $q(u)$ that is closest to $v_{kj}$ among all $u\in S$, and
return it as the answer to the candidate-center-query on
$[k,j]$. Such a $q(u)$ can be found by comparing the nodes of $S$ in
$O(\log n)$ time. Specifically,
for each pair $u$ and $u'$ in a comparison, we find among $q(u)$ and
$q(u')$ the one closer to $v_{kj}$, which can be done in constant time
by using the lowest common ancestor query in the same way as in the
proof of Lemma~\ref{lem:110}, and then we keep comparing the above closer
one to the rest of the nodes in $S$. In this way, the
candidate-center-query can be handled in $O(\log n)$ time.
\qed
\end{proof}

This proves Lemma~\ref{lem:candidate}.

\subsection{Handling the Degenerate Case and Reducing the General Case to the Vertex-Constrained Case}
\label{sec:reduction}

We have solved the vertex-constrained case problem, i.e., all locations
of $\calP$ are at vertices of $T$ and each vertex of $T$ contains at
least one location of $\calP$. Recall that we have made a general position assumption that every
vertex of $T$ has only one location of $\calP$. For the degenerate case, our algorithm still works in the same way as before with the following slight change. Consider a subtree $T(\mu)$ corresponding to a node $\mu$ of $\Upsilon$. In the degenerate case, since a vertex of $T(\mu)$ may hold multiple uncertain point locations of $\calP$, we define the size $|T(\mu)|$ to be the total number of all uncertain point locations in $T(\mu)$. In this way, the algorithm and the analysis follow similarly as before. In fact, the performance of the algorithm becomes even better in the degenerate case since the height of the decomposition tree $\Upsilon$ becomes smaller (specifically, it is bounded by $O(\log t)$, where $t$ is the number of vertices of $T$, and $t<M$ in the degenerate case).

The above has solved the vertex-constrained case problem (including the degenerate case). In the general case, a location of
$\calP$ may be in the interior of an edge of $T$ and a vertex of $T$ may not
hold any location of $\calP$. The following theorem solves the general case by
reducing it to the vertex-constrained case. The reduction is almost
the same as the one given in \cite{ref:WangCo16} for the one-center
problem and we include it here for the completeness of this paper.

\begin{lemma}\label{lem:reduction}
The center-coverage problem on $\calP$ and $T$ is solvable in
$O(\tau+M+|T|)$
time, where $\tau$ is the time for solving the same problem on $\calP$
and $T$ if this were a vertex-constrained case.
\end{lemma}
\begin{proof}
We reduce the problem to an instance of the vertex-constrained case and
then apply our algorithm for the vertex-constrained case. More specifically, we
will modify the tree $T$ to obtain another tree $T'$ of size $\Theta(M)$. We will
also compute another set $\calP'$ of $n$ uncertain points on $T'$, which
correspond to the uncertain points of $\calP$ with the same weights, but
each uncertain point $P_i$ of $\calP'$ has at most $2m_i$ locations on $T'$.
Further, each location of $\calP'$ is at a vertex of $T'$ and each vertex of $T'$ holds
at least one location of $\calP'$, i.e., it is the vertex-constrained case. We will show that we can obtain $T'$ and $\calP'$ in $O(M+|T|)$ time.
Finally, we will show that given a set of centers on $T'$ for $\calP'$, we
can find a corresponding set of the same number of centers on $T$ for
$\calP$ in $O(M+|T|)$ time.  The details are given below.

We assume that for each edge $e$ of $T$, all locations of $\calP$ on $e$ have been
sorted (otherwise we sort them first, which would introduce an additional
$O(M\log M)$ time on the problem reduction).
We traverse $T$, and for each edge $e$, if $e$ contains
some locations of $\calP$ in its interior, we create a new vertex in $T$ for
each such location. In this way, we create at most $M$ new vertices for $T$.
The above can be done in $O(M+|T|)$ time. We
use $T_1$ to denote the new tree. Note that $|T_1|=O(M+|T|)$.
For each vertex $v$ of $T_1$, if $v$ does not
hold any location of $\calP$, we call $v$ an {\em empty} vertex.

Next, we modify $T_1$ in the following way. First, for each leaf $v$ of $T_1$,
if $v$ is empty, then we remove $v$ from $T_1$. We
keep doing this until each leaf of the remaining tree is not empty.
Let $T_2$ denote the tree after the above step (e.g., see Fig.~\ref{fig:general}(b)).
Second, for each
internal vertex $v$ of $T_2$, if the degree of $v$ is $2$ and $v$ is empty,
then we remove $v$ from $T_2$ and merge its two
incident edges as a single edge whose length is equal to the sum of the lengths
of the two incident edges of $v$. We keep doing this until each degree-2 vertex of
the remaining tree is not empty. Let $T'$ represent the remaining
tree (e.g., see Fig.~\ref{fig:general}(c)).
The above two steps can be implemented in $O(|T_1|)$ time, e.g., by a
post-order traversal of $T_1$. We omit the details.

\begin{figure}[t]
\begin{minipage}[t]{\linewidth}
\begin{center}
\includegraphics[totalheight=1.3in]{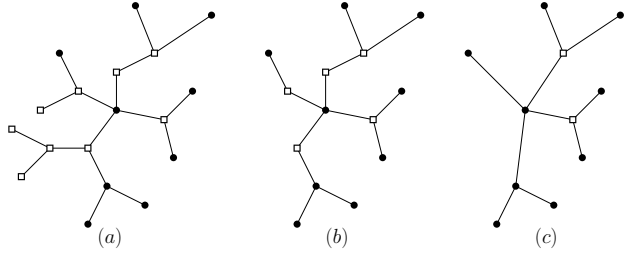}
\caption{\footnotesize Illustrating the three trees: (a) $T_1$, (b) $T_2$, and (c)
$T'$, where the empty and non-empty vertices are shown with squares and disks,
respectively.}
\label{fig:general}
\end{center}
\end{minipage}
\vspace*{-0.15in}
\end{figure}

Notice that every location of $\calP$ is at a
vertex of $T'$ and every vertex of $T'$ except those whose degrees are
at least three holds a location of $\calP$. Let $V$ denote the set of all
vertices of $T'$ and let $V_3$ denote the set of the vertices of $T'$ whose
degrees are at least three. Clearly, $|V_3|\leq |V\setminus V_3|$.
Since each vertex in $V\setminus V_3$ holds a location of $\calP$, we have
$|V\setminus V_3|\leq M$, and thus $|V_3|\leq M$.

To make every vertex of $T'$ contain a location of an uncertain point,
we first arbitrarily pick $m_1$ vertices from $V_3$ and remove them from $V_3$,
and set a ``dummy'' location
for $P_1$ at each of these vertices with zero probability. We keep picking
next $m_2$ vertices from $V_3$ for $P_2$ and continue this procedure until $V_3$
becomes empty. Since $|V_3|\leq M$, the above procedure will
eventually make $V_3$ empty before we ``use up'' all $n$ uncertain points of
$\calP$. We let $\calP'$ be the set of new uncertain points. For each $P_i\in \calP$,
it has at most $2m_i$ locations on $T'$.

Since now every vertex of $T'$ holds a location of $\calP'$ and every
location of $\calP'$ is at a vertex of $T'$, we obtain
an instance of the vertex-constrained case on $T'$ and $\calP'$.
Hence, we can use our algorithm for the vertex-constrained case to compute
a set $C'$ of centers on $T'$ in $O(\tau)$ time.
In the following, for each center $c'\in C'$, we find a corresponding
center $c$ on the original tree $T$ such that $P_i$ is
covered by $c$ on $T$ if and only if $P'_i$ is covered by $c'$ on $T'$.

Observe that every vertex $v$ of $T'$ also exists as a vertex in $T_1$, and every edge
$(u,v)$ of $T'$ corresponds to the simple path in $T_1$ between $u$ and $v$. Suppose
$c'$ is on an edge $(u,v)$ of $T'$ and let $\delta$ be the length of $e$ between $u$ and $c'$.
We locate a corresponding $c_1$ in $T_1$ in the simple path from $u$
to $v$ at distance $\delta$ from $u$.
On the other hand, by our construction from $T$ to $T_1$, if an
edge $e$ of $T$ does not appear in $T_1$, then $e$ is broken into several edges
in $T_1$ whose total length is equal to that of $e$. Hence, every
point of $T$ corresponds to a point on $T_1$. We find the point on $T$ that corresponds to $c_1$ of $T_1$, and let the point be $c$.

Let $C$ be the set of points $c$ on $T$ corresponding to
all $c'\in C'$ on $T$, as defined above. Let $C_1$ be the set of points $c_1$
on $T_1$ corresponding to all $c'\in C'$ on $T'$. To compute $C$, we
first compute $C_1$. 
This can be
done by traversing both $T'$ and $T_1$, i.e., for each edge $e$ of
$T'$ that contains centers $c'$ of $C'$, we find the corresponding
points $c_1$ in the path of $T_1$ corresponding to the edge $e$. Since
the paths of $T_1$ corresponding to the edges of $T'$ are pairwise
edge-disjoint, the runtime for computing $C_1$ is
$O(|T_1|+|T'|)$. Next we compute $C$, and similarly this can be done
by traversing both $T_1$ and $T$ in $O(|T_1|+|T|)$ time. Hence, the
total time for computing $C$ is $O(|T|+M)$ since both $|T_1|$ and $|T'|$ are bounded by $O(|T|+M)$.

As a summary, we can find an optimal solution for the
center-coverage problem on $T$ and $\calP$ in $O(\tau+M+|T|)$ time.
The lemma thus follows. \qed
\end{proof}

\section{The $k$-Center Problem}
\label{sec:kcenter}

The $k$-center problem is to find a set $C$ of $k$ centers on $T$ minimizing the value $\max_{1\leq i\leq n}d(C,P_i)$, where $d(C,P_i)=\min_{c\in C}d(c,P_i)$. Let $\lambda_{opt}=\max_{1\leq i\leq n}d(C,P_i)$ for an optimal solution $C$, and we call $\lambda_{opt}$ the {\em optimal covering range}.

As the center-coverage problem, we can also reduce the general $k$-center
problem to the vertex-constrained case. The reduction is similar to
the one in Lemma~\ref{lem:reduction} and we omit the details. In the
following, we only discuss the vertex-constrained case and we assume the problem on $T$ and $\calP$ is a vertex-constrained case.
Let $\tau$ denote the running time for solving the center-coverage
algorithm on $T$ and $\calP$.

To solve the $k$-center problem, the key is to compute
$\lambda_{opt}$, after which we can compute $k$ centers in additional
$O(\tau)$ time using
our algorithm for the center-coverage problem with $\lambda=\lambda_{opt}$.
To compute $\lambda_{opt}$, there are two main steps. In the first
step, we find a set $S$ of $O(n^2)$
{\em candidate values} such that $\lambda_{opt}$ must be in $S$. In
the second step, we compute $\lambda_{opt}$ in $S$.
Below we first compute the set $S$.

\begin{figure}[t]
\begin{minipage}[t]{\linewidth}
\begin{center}
\includegraphics[totalheight=1.0in]{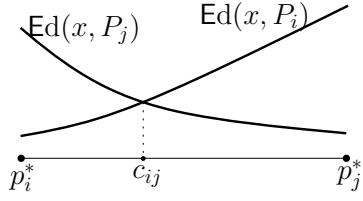}
\caption{\footnotesize Illustrating $c_{ij}$ and the two functions $\Ed(x,P_i)$ and $\Ed(x,P_j)$ as $x$ changes in the path $\pi(p_i^*,p_j^*)$ (shown as a segment).}
\label{fig:cij}
\end{center}
\end{minipage}
\vspace*{-0.15in}
\end{figure}

For any two medians $p_i^*$ and $p_j^*$ on $T_m$, observe that as $x$ moves on $\pi(p_i^*,p_j^*)$ from $p_i^*$ to $p_j^*$, $\Ed(x,P_i)$ is monotonically increasing and $\Ed(x,P_j)$ is monotonically decreasing (e.g., see Fig.~\ref{fig:cij}); we define $c_{ij}$ to be a point on
the path $\pi(p_i^*,p_j^*)$ with $\Ed(c_{ij},P_i)=\Ed(c_{ij},P_j)$, and we let $c_{ij}=\emptyset$ if such a point does not exist on $\pi(p_i^*,p_j^*)$. We have the following lemma.

\begin{lemma}\label{lem:160}
Either $\lambda_{opt}=\Ed(p_i^*,P_i)$ for some  uncertain point $P_i$ or
$\lambda_{opt}=\Ed(c_{ij},P_i)=\Ed(c_{ij},P_j)$ for two
uncertain points $P_i$ and $P_j$.
\end{lemma}
\begin{proof}
%
%
Consider any optimal solution and let $C$ be the set of all centers.
For each $c\in C$, let $Q(c)$ be the set of uncertain points that are
covered by $c$ with respect to $\lambda_{opt}$, i.e., for each $P_i\in
Q(c)$, $\Ed(c,P_i)\leq \lambda_{opt}$.
Let $C'$ be the subset of all centers $c\in C$ such that $Q(c)$ has an
uncertain point $P_i$ with $\Ed(c,P_i)=\lambda_{opt}$ and there is no
other center $c'\in C$ with $\Ed(c',P_i)< \lambda_{opt}$. For each $c\in
C'$, let $Q'(c)$ be the set of all uncertain points $P_i$ such that
$\Ed(c,P_i)=\lambda_{opt}$.

If there exists a center $c\in C'$ with an uncertain point $P_i\in
Q'(c)$ such that $c$ is at $p_i^*$, then the lemma follows since
$\lambda_{opt}=\Ed(c,P_i)=\Ed(p_i^*,P_i)$. Otherwise, if
there exists a center $c\in C'$ with two uncertain points $P_i$ and
$P_j$ in $Q'(c)$ such that $c$ is at $c_{ij}$, then the lemma also
follows since  $\lambda_{opt}=\Ed(c_{ij},P_i)=\Ed(c_{ij},P_j)$.
Otherwise, if we move each $c\in C'$ towards the median $p_j^*$ for any
$P_j\in Q'(c)$, then $\Ed(c,P_i)$ for every $P_i\in Q'(c)$ becomes
non-increasing. During the above movements of all $c\in C'$, one of
the following two cases must happen (since otherwise we would obtain another
set $C''$ of $k$ centers with $\max_{1\leq i\leq
n}d(C'',P_i)<\lambda_{opt}$, contradicting with that $\lambda_{opt}$ is the optimal covering range): either a center
$c$ of $C'$ arrives at a median $p_i^*$ with
$\lambda_{opt}=\Ed(c,P_i)=\Ed(p_i^*,P_i)$ or a center $c$ of $C'$
arrives at $c_{ij}$ for two uncertain points $P_i$ and $P_j$ with
$\lambda_{opt}=\Ed(c_{ij},P_i)=\Ed(c_{ij},P_j)$. In either case, the lemma
follows.
\qed
\end{proof}

In light of Lemma~\ref{lem:160}, we let $S=S_1\cup S_2$ with
$S_1=\{\Ed(p_i^*,P_i)\ |\ 1\leq i\leq n\}$ and
$S_2=\{\Ed(c_{ij},P_i)\ |\ 1\leq i,j\leq n\}$ (if $c_{ij}=\emptyset$ for a pair $i$ and $j$, then let $\Ed(c_{ij},P_i)=0$). Hence,
$\lambda_{opt}$ must be in $S$ and $|S|=O(n^2)$.

We assume the data structure $\calA_3$ has been computed in $O(M\log M)$
time. Then, computing the values of $S_1$ can be
done in $O(n\log M)$ time by using $\calA_3$.
The following lemma computes $S_2$ in $O(M+ n^2\log n\log M)$
time.

\begin{lemma}
After $O(M)$ time preprocessing, we can compute $\Ed(c_{ij},P_i)$ in
$O(\log n\cdot \log M)$ time for any pair $i$ and $j$.
\end{lemma}
\begin{proof}
As preprocessing, we do the following. First, we compute a lowest common ancestor query data structure on $T_m$ in $O(M)$ time such that given any two vertices of $T_m$, their lowest common ancestor can be found in $O(1)$ time~\cite{ref:BenderTh00,ref:HarelFa84}. Second, for each vertex $v$ of $T_m$, we compute the length $d(v,r)$, i.e., the number of edges in the path of $T_m$ from $v$ to the root $r$ of $T_m$. Note that $d(v,r)$ is also the depth of $v$. Computing $d(v,r)$ for all vertices $v$ of $T_m$ can be done in $O(M)$ time by a depth-first-traversal of $T_m$ starting from $r$. For each vertex $v\in T_m$ and any integer $d\in [0,d(v,r)]$, we use $\alpha(v,d)$ to denote the ancestor of $v$ whose depth is $d$. We build a {\em level ancestor query} data structure on $T_m$ in $O(M)$ time that can compute $\alpha(v,d)$ in constant time for  any vertex $v$ and any $d\in [0,d(v,r)]$~\cite{ref:BenderTh04}. The total time of the above processing is $O(M)$.

Consider any pair $i$ and $j$. We present an algorithm to
compute $c_{ij}$ in $O(\log n\cdot \log M)$ time, after which $\Ed(c_{ij},P_i)$
can be computed in $O(\log M)$ time by using the data structure $\calA_3$.

Observe that $c_{ij}\neq \emptyset$ if and only if $\Ed(p_i^*,P_i)\leq
\Ed(p_i^*,P_j)$ and $\Ed(p_j^*,P_j)\leq \Ed(p_j^*,P_i)$. Using
$\calA_3$, we can compute the four expected distances in $O(\log
M)$ time and thus determine whether $c_{ij}=\emptyset$. If yes, we simply
return zero. Otherwise, we proceed as follows.


Note that $c_{ij}$ is a point $x\in \pi(p_i^*,p_j^*)$ minimizing the value $\max\{\Ed(x,P_i),\Ed(x,P_j)\}$ (e.g., see Fig.~\ref{fig:cij}).
To compute $c_{ij}$, by using a lowest common ancestor query, we find the lowest common ancestor $v_{ij}$ of $p_i^*$ and $p_j^*$ in constant time. Then, we search $c_{ij}$ on the path $\pi(p_i^*,v_{ij})$, as follows (we will search the path $\pi(p_j^*,v_{ij})$ later). To simplify the
notation, let $\pi=\pi(p_i^*,v_{ij})$. By using the level ancestor queries, we can find the middle edge of $\pi$ in $O(1)$ time. Specifically, we find the two vertices $v_1=\alpha(p_i^*,k)$ and $v_2=\alpha(p_i^*,k+1)$, where $k=\lfloor(d(p_i^*,r)+d(v_{ij},r))/2\rfloor$. Note that the two values $d(p_i^*,r)$ and $d(v_{ij},r)$ are computed in the preprocessing. Hence, $v_1$ and $v_2$ can be found in constant time by the level ancestor queries. Clearly, the edge $e=(v_1,v_2)$ is the middle edge of $\pi$.

%
%

After $e$ is obtained, by using the data structure $\calA_3$ and as remarked in Section~\ref{sec:ed},
we can obtain the two functions $\Ed(x,P_i)$ and $\Ed(x,P_j)$ on $x\in e$ in $O(\log M)$ time, and both functions are linear in $x$ for $x\in e$. As $x$ moves in $e$ from one
end to the other, one of $\Ed(x,P_i)$ and $\Ed(x,P_j)$ is
monotonically increasing and the other is monotonically decreasing.
Therefore, we can determine in constant time whether $c_{ij}$ is on $\pi_1$,
$\pi_2$, or $e$, where $\pi_1$ and $\pi_2$ are the sub-paths of $\pi$
partitioned by $e$. If $c_{ij}$ is on $e$, then $c_{ij}$ can be
computed immediately by the two functions and we can finish the algorithm. Otherwise, the binary search
proceeds on either $\pi_1$ or $\pi_2$ recursively.

For the runtime, the binary search has $O(\log n)$ iterations and each
iteration runs in $O(\log M)$ time. So the total time of the binary search on $\pi(p_i^*,v_{ij})$ is $O(\log n\log M)$. The binary search will either find $c_{ij}$ or determine that $c_{ij}$ is at $v_{ij}$. The latter case actually implies that $c_{ij}$ is in the path $\pi(p_j^*,v_{ij})$, and thus we apply the similar binary search on $\pi(p_j^*,v_{ij})$, which will eventually compute $c_{ij}$.
Thus, the total time
for computing $c_{ij}$ is $O(\log n\log M)$.

The lemma thus follows. \qed
\end{proof}

The following theorem summarizes our algorithm.
\begin{theorem}
An optimal solution for the $k$-center problem can be found
in $O(n^2\log n\log M+M\log^2M\log n)$ time.
\end{theorem}
\begin{proof}
Assume the data structure $\calA_3$ has been computed in $O(M\log M)$
time.  Computing $S_1$ can be done in $O(n\log M)$ time. Computing
$S_2$ takes $O(M+n^2\log n\log M)$ time.
After $S$ is computed, we find $\lambda_{opt}$ from $S$ as follows.

Given any $\lambda$ in $S$, we can use  our algorithm for the
center-coverage  problem to find a minimum number $k'$ of centers with respect to
$\lambda$. If $k'\leq k$, then we say that $\lambda$ is {\em
feasible}. Clearly, $\lambda_{opt}$ is the smallest feasible value in
$S$. To find $\lambda_{opt}$ from $S$,
we first sort all values in $S$ and then do binary search
using our center-coverage algorithm as a decision procedure. In this
way, $\lambda_{opt}$ can be found in $O(n^2\log n+\tau\log n)$ time.

Finally, we can find an
optimal solution using our algorithm for the covering problem with
$\lambda=\lambda_{opt}$ in $O(\tau)$ time.
Therefore, the total time of the algorithm is
$O(n^2\log n\log M+\tau\log n)$, which is
$O(n^2\log n\log M+M\log^2M\log n)$ by Theorem~\ref{theo:10}.
\qed
\end{proof}



\begin{thebibliography}{10}

\bibitem{ref:AgarwalIn09}
P.K. Agarwal, S.-W. Cheng, Y.~Tao, and K.~Yi.
\newblock Indexing uncertain data.
\newblock In {\em Proc. of the 28th Symposium on Principles of Database Systems
  (PODS)}, pages 137--146, 2009.

\bibitem{ref:AgarwalNe12}
P.K. Agarwal, A.~Efrat, S.~Sankararaman, and W.~Zhang.
\newblock Nearest-neighbor searching under uncertainty.
\newblock In {\em Proc. of the 31st Symposium on Principles of Database Systems
  (PODS)}, pages 225--236, 2012.

\bibitem{ref:AgarwalCo14}
P.K. Agarwal, S.~Har-Peled, S.~Suri, H.~Y{\i}ld{\i}z, and W.~Zhang.
\newblock Convex hulls under uncertainty.
\newblock In {\em Proc. of the 22nd Annual European Symposium on Algorithms
  (ESA)}, pages 37--48, 2014.

\bibitem{ref:AgarwalEf98}
P.K. Agarwal and M.~Sharir.
\newblock Efficient algorithms for geometric optimization.
\newblock {\em ACM Computing Surveys}, 30(4):412--458, 1998.

\bibitem{ref:AverbakhFa05}
I.~Averbakh and S.~Bereg.
\newblock Facility location problems with uncertainty on the plane.
\newblock {\em Discrete Optimization}, 2:3--34, 2005.

\bibitem{ref:AverbakhMi97}
I.~Averbakh and O.~Berman.
\newblock Minimax regret {$p$-center} location on a network with demand
  uncertainty.
\newblock {\em Location Science}, 5:247--254, 1997.

\bibitem{ref:BenderTh00}
M.~Bender and M.~Farach-Colton.
\newblock The {LCA} problem revisited.
\newblock In {\em Proc. of the 4th Latin American Symposium on Theoretical
  Informatics}, pages 88--94, 2000.

\bibitem{ref:BenderTh04}
M.A. Bender and M.~Farach-Colton.
\newblock The level ancestor problem simplied.
\newblock {\em Theoretical Computer Science}, 321:5--12, 2004.

\bibitem{ref:BeregOp15}
S.~Bereg, B.~Bhattacharya, S.~Das, T.~Kameda, P.R.S. Mahapatra, and Z.~Song.
\newblock Optimizing squares covering a set of points.
\newblock {\em Theoretical Computer Science}, in press, 2015.

\bibitem{ref:BrodalDy02}
G.~Brodal and R.~Jacob.
\newblock Dynamic planar convex hull.
\newblock In {\em Proc. of the 43rd IEEE Symposium on Foundations of Computer
  Science (FOCS)}, pages 617--626, 2002.

\bibitem{ref:ChanGe15}
T.M. Chan and N.~Hu.
\newblock Geometric red–blue set cover for unit squares and related problems.
\newblock {\em Computational Geometry}, 48(5):380--385, 2015.

\bibitem{ref:ChazelleFr86}
B.~Chazelle and L.~Guibas.
\newblock Fractional cascading: {I. A} data structuring technique.
\newblock {\em Algorithmica}, 1(1):133--162, 1986.

\bibitem{ref:ChengCl08}
R.~Cheng, J.~Chen, and X.~Xie.
\newblock Cleaning uncertain data with quality guarantees.
\newblock {\em Proceedings of the VLDB Endowment}, 1(1):722--735, 2008.

\bibitem{ref:ChengEf04}
R.~Cheng, Y.~Xia, S.~Prabhakar, R.~Shah, and J.S. Vitter.
\newblock Efficient indexing methods for probabilistic threshold queries over
  uncertain data.
\newblock In {\em Proc. of the 30th International Conference on Very Large Data
  Bases (VLDB)}, pages 876--887, 2004.

\bibitem{ref:ColeSl87}
R.~Cole.
\newblock Slowing down sorting networks to obtain faster sorting algorithms.
\newblock {\em Journal of the ACM}, 34(1):200--208, 1987.

\bibitem{ref:BergKi13}
M.~de~Berg, M.~Roeloffzen, and B.~Speckmann.
\newblock Kinetic 2-centers in the black-box model.
\newblock In {\em Proc. of the 29th Annual Symposium on Computational Geometry
  (SoCG)}, pages 145--154, 2013.

\bibitem{ref:DongDa07}
X.~Dong, A.Y. Halevy, and C.~Yu.
\newblock Data integration with uncertainty.
\newblock In {\em Proceedings of the 33rd International Conference on Very
  Large Data Bases}, pages 687--698, 2007.

\bibitem{ref:FredericksonPa91}
G.N. Frederickson.
\newblock Parametric search and locating supply centers in trees.
\newblock In {\em Proc. of the 2nd International Workshop on Algorithms and
  Data Structures (WADS)}, pages 299--319, 1991.

\bibitem{ref:FredericksonFi83}
G.N. Frederickson and D.B. Johnson.
\newblock Finding {$k$th} paths and {$p$-centers} by generating and searching
  good data structures.
\newblock {\em Journal of Algorithms}, 4(1):61--80, 1983.

\bibitem{ref:GonzalezCo91}
T.~F. Gonzalez.
\newblock Covering a set of points in multidimensional space.
\newblock {\em Information Processing Letters}, 40(4):181--188, 1991.

\bibitem{ref:HarelFa84}
D.~Harel and R.E. Tarjan.
\newblock Fast algorithms for finding nearest common ancestors.
\newblock {\em SIAM Journal on Computing}, 13:338--355, 1984.

\bibitem{ref:HochbaumAp85}
D.S. Hochbaum and W.~Maass.
\newblock Approximation schemes for covering and packing problems in image
  processing and vlsi.
\newblock {\em Journal of the ACM}, 32(1):130--136, 1985.

\bibitem{ref:HuangSt17}
L.~Huang and J.~Li.
\newblock Stochasitc $k$-center and $j$-flat-center problems.
\newblock In {\em Proc. of the 28th Annual ACM-SIAM Symposium on Discrete
  Algorithms (SODA)}, pages 110--129, 2017.

\bibitem{ref:JorgensenGe11}
A.~J{\o}rgensen, M.~L{\"o}ffler, and J.M. Phillips.
\newblock Geometric computations on indecisive points.
\newblock In {\em Proc. of the 12nd Algorithms and Data Structures Symposium
  (WADS)}, pages 536--547, 2011.

\bibitem{ref:KamousiCl11}
P.~Kamousi, T.M. Chan, and S.~Suri.
\newblock Closest pair and the post office problem for stochastic points.
\newblock In {\em Proc. of the 12nd International Workshop on Algorithms and
  Data Structures (WADS)}, pages 548--559, 2011.

\bibitem{ref:KamousiSt11}
P.~Kamousi, T.M. Chan, and S.~Suri.
\newblock Stochastic minimum spanning trees in {Euclidean} spaces.
\newblock In {\em Proc. of the 27th Annual Symposium on Computational Geometry
  (SoCG)}, pages 65--74, 2011.

\bibitem{ref:KarivAn79}
O.~Kariv and S.~Hakimi.
\newblock An algorithmic approach to network location problems. {II: The}
  $p$-medians.
\newblock {\em SIAM Journal on Applied Mathematics}, 37(3):539--560, 1979.

\bibitem{ref:KarivAnC79}
O.~Kariv and S.L. Hakimi.
\newblock An algorithmic approach to network location problems. {I: The}
  $p$-centers.
\newblock {\em SIAM J. on Applied Mathematics}, 37(3):513--538, 1979.

\bibitem{ref:KimCo11}
S.-S. Kim, S.W. Bae, and H.-K. Ahn.
\newblock Covering a point set by two disjoint rectangles.
\newblock {\em International Journal of Computational Geometry and
  Applications}, 21:313--330, 2011.

\bibitem{ref:LofflerLa10}
M.~L{\"o}ffler and M.~van Kreveld.
\newblock Largest bounding box, smallest diameter, and related problems on
  imprecise points.
\newblock {\em Computational Geometry: Theory and Applications},
  43(4):419--433, 2010.

\bibitem{ref:MegiddoLi83}
N.~Megiddo.
\newblock Linear-time algorithms for linear programming in {$R^3$} and related
  problems.
\newblock {\em SIAM Journal on Computing}, 12(4):759--776, 1983.

\bibitem{ref:MegiddoNe83}
N.~Megiddo and A.~Tamir.
\newblock New results on the complexity of $p$-centre problems.
\newblock {\em SIAM Journal on Computing}, 12(4):751--758, 1983.

\bibitem{ref:MegiddoAn81}
N.~Megiddo, A.~Tamir, E.~Zemel, and R.~Chandrasekaran.
\newblock An {$O(n \log^2 n)$} algorithm for the $k$-th longest path in a tree
  with applications to location problems.
\newblock {\em SIAM J. on Computing}, 10:328--337, 1981.

\bibitem{ref:MustafaPT09}
N.H. Mustafa and S.~Ray.
\newblock {PTAS} for geometric hitting set problems via local search.
\newblock In {\em Proc. of the 25th Annual Symposium on Computational Geometry
  (SoCG)}, pages 17--22, 2009.

\bibitem{ref:SuriOn14}
S.~Suri and K.~Verbeek.
\newblock On the most likely voronoi diagram and nearest neighbor searching.
\newblock In {\em Proc. of the 25th International Symposium on Algorithms and
  Computation (ISAAC)}, pages 338--350, 2014.

\bibitem{ref:SuriOn13}
S.~Suri, K.~Verbeek, and H.~Y{\i}ld{\i}z.
\newblock On the most likely convex hull of uncertain points.
\newblock In {\em Proc. of the 21st European Symposium on Algorithms (ESA)},
  pages 791--802, 2013.

\bibitem{ref:TaoRa07}
Y.~Tao, X.~Xiao, and R.~Cheng.
\newblock Range search on multidimensional uncertain data.
\newblock {\em ACM Transactions on Database Systems}, 32, 2007.

\bibitem{ref:WangMi14}
H.~Wang.
\newblock Minmax regret 1-facility location on uncertain path networks.
\newblock {\em European Journal of Operational Research}, 239:636--643, 2014.

\bibitem{ref:WangOn15}
H.~Wang and J.~Zhang.
\newblock One-dimensional {$k$}-center on uncertain data.
\newblock {\em Theoretical Computer Science}, 602:114--124, 2015.

\bibitem{ref:WangCo16}
H.~Wang and J.~Zhang.
\newblock Computing the center of uncertain points on tree networks.
\newblock {\em Algorithmica}, 78(1):232--254, 2017.

\bibitem{ref:YiuEf09}
M.L. Yiu, N.~Mamoulis, X.~Dai, Y.~Tao, and M.~Vaitis.
\newblock Efficient evaluation of probabilistic advanced spatial queries on
  existentially uncertain data.
\newblock {\em IEEE Transactions on Knowledge and Data Engineering},
  21:108--122, 2009.

\end{thebibliography}

\end{document}